\documentclass[acmsmall,nonacm,authorversion,screen]{acmart}
\copyrightyear{\the\year}

\acmPrice{}
\acmDOI{10.1145/3434316}
\acmYear{2021}
\acmSubmissionID{popl21main-p234-p}
\acmJournal{PACMPL}
\acmVolume{5}
\acmNumber{POPL}
\acmMonth{1}

\setcopyright{ccby}

\usepackage[utf8]{inputenc}
\usepackage[T1]{fontenc}

\usepackage{dblfloatfix}
\usepackage{amsmath}
\usepackage{amsthm}
\usepackage{bbm}
\usepackage{stmaryrd}
\usepackage{mathtools}
\usepackage{mathpartir}
\usepackage{enumitem}
\usepackage{cancel}
\usepackage{pl-syntax}
\usepackage{xspace}
\usepackage{doi}
\usepackage{ifthen}
\usepackage{suffix}
\usepackage{microtype}
\usepackage{calc}
\usepackage{tikz}
\usepackage{footmisc}

\makeatletter
\newcommand{\citeposessive}[2][]{%
  \def\@yearlist{#2\ifx\empty#1\else,#1\fi}%
  \citeauthor{#2}'s~[\citeyear{\@yearlist}]}
\makeatother

\newcommand{\programfont}[1]{\ensuremath{\mathsf{#1}}}

\newcounter{numlevels}
\newcommand{\newmultilevelprogram}[3][0]{\newcommand{#2}[#1]{\ifthenelse{\value{numlevels} > 0}{\begin{array}[t]{@{}l@{}}}{\begin{array}{l}}\addtocounter{numlevels}{1}#3\addtocounter{numlevels}{-1}\end{array}}}

\newcommand{\newcommenter}[3]{%
  \newcommand{#1}[1]{%
    \textcolor{#2}{\small\textsf{[#3: {##1}]}}%
  }%
}
\definecolor{darkgreen}{rgb}{0,0.7,0}
\newcommenter{\ethan}{darkgreen}{Ethan}
\newcommenter{\akh}{purple}{AKH}

\theoremstyle{theorem}
\newtheorem{thm}{Theorem}

\newtheorem{lem}{Lemma}

\theoremstyle{definition}
\newtheorem{defn}{Definition}

\newcommand{\one}{\textit{i}\xspace}
\newcommand{\two}{\textit{ii}\xspace}
\newcommand{\three}{\textit{iii}\xspace}

\newcommand{\proves}{\vdash}
\newcommand{\pc}{\ensuremath{\mathrm{pc}}\xspace}

\newcommand{\Ptd}{\textrm{ptd}\xspace}

\newcommand{\ctxsep}{\diamond}

\newcommand{\provesEff}{\proves_{\mkern-4mu\raisebox{-0.5pt}{$\scriptstyle\textsc{e}$}}}
\newcommand{\provesPure}{\proves_{\mkern-4mu\raisebox{-0.5pt}{$\scriptstyle\textsc{p}$}}}
\newcommand{\provesPc}{\proves_{\mkern-6mu\raisebox{-0.5pt}{$\scriptstyle\textsc{pc}$}}}
\definecolor{eff-color}{rgb}{0.5,0.3,0}
\definecolor{pure-color}{rgb}{0,0,0.67}
\definecolor{pc-color}{rgb}{0,0.5,0}
\newcommand{\effColor}[1]{{\color{eff-color}#1}}
\newcommand{\pureColor}[1]{{\color{pure-color}#1}}
\newcommand{\pcColor}[1]{{\color{pc-color}#1}}

\newcommand{\Labs}{\ensuremath{\mathcal{L}}\xspace}
\newcommand{\flowsto}{\sqsubseteq}
\newcommand{\nflowsto}{\not\sqsubseteq}
\newcommand{\JoinL}{\sqcup}

\newcommand{\lequiv}[1][\ell]{\ensuremath{\mathrel{\approx_{#1}}}}
\newcommand{\lequivEff}[2][\ell]{\ensuremath{\mathrel{\cong^{#2}_{#1}}}}
\newcommand{\StExnLequiv}[1][\ell]{\lequivEff[#1]{\text{\upshape\scshape se}}}
\newcommand{\tslequiv}[1][\ell]{\lequivEff[#1]{\text{\upshape\scshape ts}}}

\newcommand{\Eff}{\ensuremath{\mathcal{E}}\xspace}

\newcommand{\ComposeERaw}[2]{#1 \leq #2}
\newcommand{\ComposeE}[2]{\ComposeERaw{[#1]}{#2}}

\newcommand{\MapE}[1]{\programfont{map}_{#1}}

\newcommand{\CoerceName}{\programfont{coerce}}
\newcommand{\CoerceE}[2]{\CoerceName_{#1 \mapsto #2}}
\newcommand{\JoinName}{\programfont{join}\xspace}
\newcommand{\JoinE}[2]{\JoinName_{[#1],#2}}

\newcommand{\prot}{\mathrel{\triangleleft}}

\newcommand{\pcto}[1][\pc]{\mathchoice%
  {\xrightarrow{#1}}
  {\xrightarrow{\raisebox{-1pt}[0pt][0pt]{$\scriptstyle #1$}}}
  {\xrightarrow{#1}}
  {\xrightarrow{#1}}}
\newcommand{\effto}[1][\varepsilon]{\mathchoice%
  {\xrightarrow{#1}}
  {\xrightarrow{\raisebox{-1pt}[0pt][0pt]{$\scriptstyle #1$}}}
  {\xrightarrow{#1}}
  {\xrightarrow{#1}}}

\newcommand{\CapTypePref}[1]{P_{#1}}
\newcommand{\CapturedType}[2]{\CapTypePref{#1}(#2)}
\newcommand{\pureTypeTrans}[1]{\llparenthesis #1 \rrparenthesis}
\newcommand{\CapturedPureTypeTrans}[2]{\CapTypePref{#1}\pureTypeTrans{#2}}
\newcommand{\effTypeTrans}[1]{\llbracket #1 \rrbracket}

\newcommand{\ty}{\,{:}\,}

\newcommand{\seq}{\mathbin{;}}
\newcommand{\subst}[3]{#1 [#2 \mapsto #3]}

\newcommand{\Inl}{\programfont{inl}\xspace}
\newcommand{\Inr}{\programfont{inr}\xspace}
\newcommand{\Unit}{\programfont{unit}\xspace}

\newcommand{\As}{\programfont{as}\xspace}
\newcommand{\In}{\programfont{in}\xspace}
\newcommand{\Match}{\programfont{match}\xspace}
\newcommand{\With}{\programfont{with}\xspace}
\newcommand{\End}{\programfont{end}\xspace}

\newcommand{\Unlabel}{\programfont{unlabel}\xspace}
\newcommand{\If}{\programfont{if}\xspace}
\newcommand{\Then}{\programfont{then}\xspace}
\newcommand{\Else}{\programfont{else}\xspace}
\newcommand{\Catch}{\programfont{catch}\xspace}

\newcommand{\Try}{\programfont{try}\xspace}
\newcommand{\Read}{\programfont{read}\xspace}
\newcommand{\Write}{\programfont{write}\xspace}
\newcommand{\Throw}{\programfont{throw}\xspace}
\newcommand{\Let}{\programfont{let}\xspace}
\newcommand{\Labeled}{\programfont{label}\xspace}

\newcommand{\Captured}[2]{\mathchoice%
  {\left\lfloor #2 \right\rfloor_{#1}}
  {\lfloor #2 \rfloor_{#1}}
  {\lfloor #2 \rfloor_{#1}}
  {\lfloor #2 \rfloor_{#1}}}
\WithSuffix\newcommand\Captured*[1]{\Captured{}{\begin{array}{@{}c@{}}#1\end{array}}}

\makeatletter
\newlength{\@progEqLen}
\newcommand{\progEq}[2][]{%
  \setlength{\@progEqLen}{\maxof{\widthof{\ensuremath{\scriptstyle#1}}}{\widthof{\ensuremath{\scriptstyle#2}}}}
  \addtolength{\@progEqLen}{6pt}
  \mathrel{\overset{#1}{\underset{#2}{\raisebox{1pt}{\rule{\@progEqLen}{2\arrayrulewidth}}\hspace{-\@progEqLen}\raisebox{3pt}{\rule{\@progEqLen}{2\arrayrulewidth}}}}}
}
\makeatother

\newcommand{\LabeledOp}[1]{\Labeled_{#1}}
\newcommand{\LabeledProgram}[2]{\LabeledOp{#1}\mkern2mu(#2)}
\newcommand{\LabeledType}[2]{L_{#1}\mkern2mu(#2)}

\newcommand{\Pair}[2]{({#1}, {#2})}

\newcommand{\Proj}[1]{\programfont{proj}_{#1}\xspace}

\newcommand{\TryCatch}[2]{\Try~\{#1\} \mathrel{\Catch} \{#2\}}

\newmultilevelprogram[5]{\MatchSum}{%
    \Match~#1~\With\\%
    \hspace*{0.5em}|~\Inl(#2) \mathrel{\Rightarrow} #3\\%
    \hspace*{0.5em}|~\Inr(#4) \mathrel{\Rightarrow} #5\\%
    \End%
}
\WithSuffix\newcommand\MatchSum*[5]{\Match~{#1}~\With\mid\Inl({#2}) \Rightarrow {#3} \mid \Inr({#4}) \Rightarrow {#5}~\End}

\newcommand{\UnlabelOp}[2]{\Unlabel_{#1}\mkern2mu(#2)}
\WithSuffix\newcommand\UnlabelOp*[1]{\Unlabel_{#1}}

\newmultilevelprogram[3]{\UnlabelAs}{%
  \Unlabel~#1 \mathrel{\As} #2\\%
  \In~#3%
}
\WithSuffix\newcommand\UnlabelAs*[3]{\Unlabel~#1 \mathrel{\As} #2 \mathrel{\In} #3}

\newmultilevelprogram[2]{\UnlabelIn}{%
  \Unlabel~#1\\
  \In~#2%
}
\WithSuffix\newcommand\UnlabelIn*[2]{\Unlabel~#1 \mathrel{\In} #2}

\newmultilevelprogram[3]{\LetIn}{%
  \Let~#1 = #2\\%
  \In~#3%
}
\WithSuffix\newcommand\LetIn*[3]{\Let~#1 = #2 \mathrel{\In} #3}

\newcommand{\stepsone}{\longrightarrow}
\newcommand{\stepsmany}{\stepsone^\ast}

\newcommand{\Public}{\programfont{public}\xspace}
\newcommand{\Secret}{\programfont{secret}\xspace}

\newcommand{\Fix}{\programfont{fix}\xspace}
\newcommand{\Lift}{\programfont{lift}\xspace}
\newcommand{\Seq}{\programfont{seq}\xspace}

\newcommand{\StateLabel}{\ensuremath{\ell_{\programfont{State}}}\xspace}
\newcommand{\ExnLabel}{\ensuremath{\ell_{\programfont{Exn}}}\xspace}
\newcommand{\TermLabel}{\ensuremath{\ell_{\programfont{PNT}}}\xspace}
\newcommand{\AtkLabel}{\ensuremath{\ell_{\programfont{Atk}}}\xspace}

\newcommand{\e}{\varepsilon}

\newcommand{\StateEff}{\ensuremath{\mathrm{S}}\xspace}
\newcommand{\REff}{\ensuremath{\mathrm{R}}\xspace}
\newcommand{\WEff}{\ensuremath{\mathrm{W}}\xspace}
\newcommand{\ExnEff}{\ensuremath{\mathrm{E}}\xspace}
\newcommand{\PntEff}{\ensuremath{\mathrm{PNT}}\xspace}

\newcommand{\Monad}[1]{\ensuremath{P_{#1}}}

\newcommand{\LiftType}[1]{{#1}_{\bot}}
\newcommand{\LiftProgram}[1]{\Lift(#1)}

\newcommand{\BindM}[1]{\programfont{bind}_{#1}\xspace}

\newcommand{\Fixpoint}[3]{\Fix~#1 \ty #2.\,#3}

\newmultilevelprogram[3]{\SeqTerm}{%
  \Seq~#1 = #2\\%
  \In~#3%
}
\WithSuffix\newcommand\SeqTerm*[3]{\Seq~#1 = #2 \mathrel{\In} #3}

\newmultilevelprogram[3]{\ITE}{%
    \If~#1\\%
    \Then~#2\\%
    \Else~#3%
}
\WithSuffix\newcommand\ITE*[3]{\If~#1~\Then~#2~\Else~#3}

\title{Giving Semantics to Program-Counter Labels via Secure~Effects}

\author{Andrew K. Hirsch}
\orcid{0000-0003-2518-614X}
\affiliation{
  \institution{Max Planck Institute for Software Systems}
  \city{Kaiserslautern and Saarbr\"ucken}
  \country{Germany}
}
\email{akhirsch@mpi-sws.org}

\author{Ethan Cecchetti}
\affiliation{
  \institution{Cornell University}
  \city{Ithaca}
  \state{New York}
  \country{USA}
}
\email{ethan@cs.cornell.edu}
\date{}

\begin{document}


\begin{abstract}
  Type systems designed for information-flow control commonly use a \emph{program-counter label} to track the sensitivity of the context and rule out data leakage arising from effectful computation in a sensitive context.
Currently, type-system designers reason about this label informally except in security proofs, where they use ad-hoc techniques.
We develop a framework based on monadic semantics for effects to give semantics to program-counter labels.
This framework leads to three results about program-counter labels.
First, we develop a new proof technique for noninterference, the core security theorem for information-flow control in effectful languages.
Second, we unify notions of security for different types of effects, including state, exceptions, and nontermination.
Finally, we formalize the folklore that program-counter labels are a lower bound on effects.
We show that, while not universally true, this folklore has a good semantic foundation.

\end{abstract}

\maketitle

\section{Introduction}
\label{sec:introduction}

Static information-flow control~(IFC) assigns \emph{information-flow labels} to data within a program.
These labels describe the sensitivity of the data.
For instance, data labeled secret is more sensitive than data labeled public.
The type system then prevents more-sensitive inputs from influencing less-sensitive outputs.
A flow of information can be \emph{explicit}, if a program directly returns an input,
or \emph{implicit} if a program conditions on the input and returns a different value from each branch.
In both cases, the type system can enforce \emph{noninterference}---the powerful safety property that a program's sensitive inputs will not influence its less-sensitive outputs~\citep{GoguenM82}---%
by checking that the output is at least as sensitive as the inputs used to compute it.

When we combine effects with implicit flows, however, this simple output checking becomes insufficient.
\citet{VolpanoSI96} demonstrate this concern with the following simple program where the secret value~$x$ is either 0 or 1 and \Write modifies the state and returns the singleton value~$()$ of type~\Unit:
$$\ITE*{x = 1}{\Write(1)}{\Write(0)}$$
While this program does not directly write a secret and always returns the same thing, an attacker who can read the final state now learns the value of $x$.

Languages often rule out these effectful implicit flows by tracking the sensitivity of the current control-flow with a \emph{program-counter label} (written~\pc) in the typing judgment~\citep[e.g.,][]{Myers99,PottierS02,MilanoM18}.
If the \pc is private, then private data influenced which command is executing, so writing to public outputs may leak that data.
If the \pc is public, however, only public data has determined which program path was taken, so a decision to write leaks nothing.
For instance, a type system with a program-counter label can detect the leak above since we write to state after branching on secret data.

Type-system designers commonly use such intuitive reasoning when building their type system.
Then they adjust the type system as needed to prove noninterference.
Ideally, designers would instead use semantically- and mathematically-grounded design principles to design type systems.
This approach would make the proof of noninterference almost trivial, since the mathematical grounding of the design principles would guarantee noninterference.
Developing such design principles requires a semantic model of program-counter labels.

A piece of folklore gives a clue for how to develop these semantic models: the \pc label is a lower bound on the effects that can occur in a well-typed program.
Taken literally, this folklore does not even seem to type-check since effects are not labels.
However, it suggests that we need a framework that relates effects and labels in a meaningful way.

To investigate this intuition, we employ a common semantic model for effects.
We translate the earlier example to a monadic form, which returns a pair consisting of the original output and the state set by \Write.
$$\ITE*{x = 1}{\Pair{()}{0}}{\Pair{()}{1}}$$
Indeed, after this translation, checking only the output is sufficient to detect any leaks.
The \pc label is no longer necessary, lending credence to the above-mentioned piece of folklore.

We formalize these intuitions by building a semantic model of program-counter labels based on monadic treatments of effects.
We base our framework on a categorical construct called a \emph{productor}~\citep{Tate13}.
Productors provide the most-general-known framework for the semantics of \emph{producer effects}---a generalization of monadic effects.
(In fact, \citet{Tate13} argues that productors are the most general possible framework for producer effects.)

Since productors, like monads, are a categorical construct, naively applying them to a programming language would require that the programming language only have one variable in its context.
We circumvent this weakness by following a suggestion from \citeposessive{Tate13} conclusion and developing \emph{strong} productors, allowing us to apply our framework to realistic languages.

Our framework requires that the productors capture the sensitivity of the effects they encode.
For instance, when translating the above example into monadic form, the left side of the output pair must capture the visibility of the old output, while the right side must capture the visibility of the heap.
If the translated program were well-typed in a noninterfering language---which it was not in the insecure example above---we would therefore be sure that the original program did not leak data.
We refer to effects captured by these security-typed productors as \emph{secure effects}.

Our core theorem enables proofs of noninterference for effectful languages (with \pc labels) that fit our framework
while only proving it directly for the pure part of the languages (without \pc labels).
As far as we are aware, this is the first theorem proving noninterference for a large swath of languages.
Moreover, this style of proof is nearly unknown in the literature.
(\citet{AlgehedR17} mention that it is possible, but do not explore it in any depth.)

For languages that fit our framework, proving noninterference (of the effects) is almost trivial, as expected.
However, to fit our framework, a language's effects must be secure, and showing that an effect is secure---that it has a productor that properly captures its visibility---%
requires reasoning about who can see the results of the effect.
Luckily, for important examples, this reasoning is not difficult, so our proof technique leads to simpler proofs than previous techniques.
As a result, we call this proof technique \emph{Noninterference Half-Off}.

In addition to this new proof technique, we use our framework to unify different notions of noninterference from the information-flow literature.
Some notions consider the termination behavior of programs (termination-sensitive), while others do not (termination-insensitive).
Our framework shows that these two notions of noninterference are distinguished by whether or not nontermination is considered a secure effect.
This view is both intellectually satisfying and provides half-off proofs of termination-sensitive noninterference.

Finally, we formalize the folklore that \pc labels serve as a lower bound on effects.
We show that the aphorism is not always true for noninterfering languages that fit our framework,
but it is true for fundamental reasons in every realistic information-flow language of which we are aware.

In Section~\ref{sec:pure-ifc-lang}, we review the Dependency Core Calculus~(DCC)~\citep{AbadiBJHR99}, a simple, pure, noninterfering language.
DCC serves as an introduction to necessary parts of IFC languages and as the basis of our example languages throughout the paper.
We then add the following contributions:
\begin{itemize}
\item We demonstrate a productor-based translation for a language with state and exceptions.
  Beyond exploring the semantics of secure effects, this allows a simple proof of noninterference (Section~\ref{sec:state-and-exns}).
\item By treating possible nontermination as an effect---as is common in the effects literature---we obtain a simple proof of termination-sensitive noninterference (Section~\ref{sec:tsni}).
\item We present our general semantic framework for effectful languages with IFC labels (Section~\ref{sec:prod-effects-ifc}), allowing us to prove properties about a wide class of IFC languages.
\item We define and prove the \emph{Noninterference Half-Off Theorem} (Theorem~\ref{thm:niho}), allowing us to extend noninterference from pure languages to effectful languages in many settings (Section~\ref{sec:niho}).
\item We show that the folklore ``the program-counter label is a lower bound on the effects in a program'' need not hold in our framework.
  We also show that extending our framework with a few simple rules makes it hold (Section~\ref{sec:deepening-pc-effect}).
\item We extend the theory of productors to include multiple-input languages like simply-typed $\lambda$-calculus and DCC (Appendix~\ref{sec:multi-input}).
\end{itemize}

\section{An Information-Flow-Control Type System for a Pure Language}
\label{sec:pure-ifc-lang}

We begin by reviewing \citeposessive{AbadiBJHR99} \emph{Dependency Core Calculus}~(DCC), a pure language with a simple noninterference property.
DCC will form the basis of our examples in Sections~\ref{sec:state-and-exns} and~\ref{sec:tsni}.
It also serves as a good language to introduce information-flow control~(IFC) and noninterference, as well as the notation for this paper.

\begin{figure*}[b]
  \vspace{-\baselineskip}
  \begin{syntax}
    \categoryFromSet[Labels]{\ell}{\Labs}
    \category[Types]{\tau} \alternative{\Unit} \alternative{\tau_1 + \tau_2} \alternative{\tau_1 \times \tau_2} \alternative{\tau_1 \to \tau_2} \alternative{\LabeledType{\ell}{\tau}}
    \category[Values]{v} \alternative{()} \alternative{\Inl(v)} \alternative{\Inr(v)} \alternative{\Pair{v_1}{v_2}} \alternative{\lambda x\ty\tau.\, e} \alternative{\LabeledProgram{\ell}{v}}
    \category[Expressions]{e} \alternative{x} \alternative{()}
    \alternative{\lambda x \ty \tau.\,e} \alternative{e_1~e_2}
    \alternative{\Pair{e_1}{e_2}} \alternative{\Proj{1}(e)} \alternative{\Proj{2}(e)}
    \alternativeLine{\Inl(e)} \alternative{\Inr(e)} \alternative{(\MatchSum*{e}{x}{e_1}{y}{e_2})}
    \alternativeLine{\LabeledProgram{\ell}{e}} \alternative{\UnlabelAs*{e_1}{x}{e_2}}
    \category[Evaluation Contexts]{E} \alternative{[\cdot]}
    \alternative{E~e} \alternative{v~E}
    \alternative{\Pair{E}{e}} \alternative{\Pair{v}{E}} \alternative{\Proj{1}(E)} \alternative{\Proj{2}(E)}
    \alternativeLine{\Inl(E)} \alternative{\Inr(E)} \alternative{(\MatchSum*{E}{x}{e_1}{y}{e_2})}
    \alternativeLine{\LabeledProgram{\ell}{E}} \alternative{\UnlabelAs*{E}{x}{e}}
  \end{syntax}
  \caption{Grammar for DCC Types and Terms}
  \label{fig:dcc-types-terms}
\end{figure*}

Figure~\ref{fig:dcc-types-terms} contains the syntax of DCC.
The heart of DCC is the simply-typed $\lambda$-calculus with products and sums.
The only additional terms are the security features that make DCC interesting from our perspective: $\LabeledProgram{\ell}{e}$ and $\UnlabelAs*{e_1}{x}{e_2}$.
We will also make free use of \Let notation, with its standard definition.
(For simplicity, we omit the fixpoint operator present in the original language~\citep{AbadiBJHR99}, though we will add it back in Section~\ref{sec:tsni}.)

The security terms use a set of \emph{information-flow labels}, \Labs, over which DCC is parameterized, that represent restrictions on data use.
For instance, if we have labels~\Secret and~\Public, then data labeled \Secret should not be used to compute data labeled \Public.
We require that labels form a preorder.
That is, there is a reflexive and transitive relation $\flowsto$ (pronounced ``flows to'').
For presentation clarity, we also assume that $\Labs$ forms a join semilattice,
meaning any two labels~$\ell_1$ and~$\ell_2$ have a join---a least upper bound---denoted $\ell_1 \JoinL \ell_2$, and there is a \emph{top} element, denoted $\top$, such that $\ell \flowsto \top$ for all~$\ell \in \Labs$.
We note again that this is only for clarity of presentation; we could replace every join with any upper bound, and disallow rules that use a join when no upper bound exists.
Intuitively, if $\ell_1 \flowsto \ell_2$, then $\ell_2$ is at least as restrictive as $\ell_1$, so $\top$ is the most-restrictive label.
We note that most IFC work assumes that labels form a lattice, meaning labels also have greatest lower bounds and there is a least element $\bot$.
We omit this additional structure as we do not find it helpful.

The term $\LabeledProgram{\ell}{e}$ represents protecting the output of~$e$ at label~$\ell$.
That is, $e$~should only be used to compute information at levels at least as high as $\ell$.
Such computations are possible using the term $\UnlabelAs*{e_1}{x}{e_2}$, which requires the output type of $e_2$ to be at a high-enough level, and if it is, allows use of $e_1$ as if it were not labeled through the variable $x$. 

The concept of a type~$\tau$ being ``of high enough level'' to use information at label~$\ell$ is expressed in a relation $\ell \prot \tau$, which is read as ``$\ell$ is protected by $\tau$'' or ``$\tau$ protects $\ell$.''
The formal rules defining this relation are in Figure~\ref{fig:dcc-protection}.
Intuitively, if $\ell \prot \tau$, then $\tau$~information is at least as secret as~$\ell$.

\begin{figure}
  \begin{mathpar}
    \infer{\ell \flowsto \ell'}{\ell \prot \LabeledType{\ell'}{\tau}}
    \and
    \infer{\ell \prot \tau}{\ell \prot \LabeledType{\ell'}{\tau}}
    \and
    \infer{\ell \prot \tau_1 \\ \ell \prot \tau_2}{\ell \prot \tau_1 \times \tau_2}
    \and
    \infer{\ell \prot \tau_2}{\ell \prot \tau_1 \to \tau_2}
  \end{mathpar}
  \caption{Protection Rules for DCC}
  \label{fig:dcc-protection}
\end{figure}

The typing rules for $\LabeledProgram{\ell}{e}$ and $\UnlabelAs*{e_1}{\ell}{e_2}$ are as follows:
\begin{mathpar}
  \infer{\Gamma \proves e : \tau}{\Gamma \proves \LabeledProgram{\ell}{e} : \LabeledType{\ell}{\tau}} \and
  \infer{\Gamma \proves e_1 : \LabeledType{\ell}{\tau_1}\\ \Gamma, x \ty \tau_1 \proves e_2 : \tau_2\\ \ell \prot \tau_2}{\Gamma \proves \UnlabelAs*{e_1}{x}{e_2} : \tau_2}
\end{mathpar}
Notice the use of the protection relation in the latter rule.
Since \Unlabel allows a program to compute with labeled data, this check requires the output of that computation to be at least as sensitive as the input.
This protection premise is the main security check in DCC's type system.

The operational semantics of DCC are mostly the standard semantics of call-by-value simply-typed $\lambda$-calculus,
so we only discuss the semantics of the terms $\LabeledProgram{\ell}{e}$ and \mbox{$\UnlabelAs*{e_1}{\ell}{e_2}$}.
We first note that $\LabeledProgram{\ell}{E}$ and $\UnlabelAs*{E}{\ell}{e_2}$ are evaluation contexts, where $E$ stands for an arbitrary evaluation context.
That is, computation can take place under both $\LabeledProgram{\ell}{-}$ and $\UnlabelAs*{-}{x}{e_2}$.
Note that computation \emph{cannot} take place in the second expression of an \Unlabel{} operation, since this is expected to run after binding the variable~$x$.
The only remaining operational semantics rule is as follows:
$$\UnlabelAs*{(\LabeledProgram{\ell}{v})}{x}{e} \stepsone \subst{e}{x}{v}$$
The full definition of the type system and operational semantics of DCC are in Appendix~\ref{sec:full-pure-typing}.

DCC's main security theorem is its \emph{noninterference} theorem.
It formalizes the fact that programs do not compute, e.g., public information with secret data.
The theorem requires a notion of equivalence at a label $\ell$, representing what an attacker who can see values only up to label $\ell$ can distinguish.
The definition is contextual to allow for comparison of first-class functions.
\begin{defn}[$\ell$-Equivalent Programs]
  \label{defn:ell-equiv}
  We say programs $e_1$ and $e_2$ are \emph{$\ell$-equivalent}, denoted $e_1 \lequiv e_2$, if for all expression contexts $C$ such that
  $\proves C[e_i] : \LabeledType{\ell}{\Unit + \Unit}$ and $C[e_i] \stepsmany v_i$ for both $i = 1, 2$, then $v_1 = v_2$.
\end{defn}

Intuitively, expressions are $\ell$-equivalent if no well-typed decision procedure with output labeled~$\ell$ can distinguish them.
Note that this definition only requires equivalent outputs when both programs terminate.
Because we omitted DCC's fixpoint operator, the language is strongly normalizing so both terms always converge.
We will address potential nontermination in Section~\ref{sec:tsni}.

We use this definition to say that two well-typed expressions must be $\ell$-equivalent unless their labels allow them to influence~$\ell$.
Formally, the protection relation defines the label of data, leading to the following theorem.
\citet{BowmanA15} proved the version we use here and \citet{AlgehedB19} provided a machine-checked proof in Agda.
\begin{thm}[Noninterference for DCC~\citep{BowmanA15,AlgehedB19}]
  \label{thm:dcc-ni}
  For expressions $e_1$ and $e_2$ and $\ell \in \Labs$, if \mbox{$\Gamma \proves e_1 : \tau$}, \mbox{$\Gamma \proves e_2 : \tau$},
  and $\ell \prot \tau$, then for all labels $\AtkLabel \in \Labs$, either $\ell \flowsto \AtkLabel$ or $e_1 \lequiv[\AtkLabel] e_2$.
\end{thm}

\section{Example: Noninterference in a Language with State and Exceptions}
\label{sec:state-and-exns}

We now extend our simple noninterfering pure language with two effects: state and exceptions.
Both are simplified for space and ease of understanding.
Specifically, we consider only one (typed) state cell and one type of exception.
Moreover, the type of this state cell,~$\sigma$, must not contain a function as a subterm in order to avoid concerns around higher-order state.
\footnote{Allowing higher-order state is possible, but it requires recursive types and complicates reasoning about termination.}
However, neither is difficult to broaden to more-realistic versions.
We extend the language with the following syntax:
\begin{syntax}
  \category[Expressions]{e} \alternative{\cdots} \alternative{\Read} \alternative{\Write(e)} \alternative{\Throw} \alternative{\TryCatch{e_1}{e_2}}
  \category[Evaluation Contexts]{E} \alternative{\cdots} \alternative{\Write(E)} \alternative{\TryCatch{E}{e_2}}
  \category[Throw Contexts]{T} \alternative{[\cdot]} \alternative{T~e} \alternative{v~T} \alternative{\Pair{T}{e}} \alternative{\Pair{v}{T}} \alternative{\Proj{1}(T)} \alternative{\Proj{2}(T)}
  \alternativeLine{\Inl(T)} \alternative{\Inr(T)} \alternative{(\MatchSum*{T}{x}{e_1}{y}{e_2})}
  \alternativeLine{\LabeledProgram{\ell}{T}} \alternative{\UnlabelAs*{T}{x}{e}} \alternative{\Write(T)}
\end{syntax}
Here \Read returns the value of type~$\sigma$ currently stored in the one state cell, while $\Write(e)$ replaces that value with~$e$ and returns $\Unit$.
We include $\Write(E)$ as an evaluation context, ensuring that~$e$ is reduced to a value before being stored.
The term \Throw throws an exception, which propagates through contexts until it either hits top-level or a \Try block.
We use a throw context $T$ to implement this propagation.
$T$ is identical to evaluation contexts except it does not include \Try blocks.
Finally, \mbox{$\TryCatch{e_1}{e_2}$} runs~$e_1$ until it returns either a value or an exception.
If it returns a value~$v$, then the try-catch returns $v$ as well.
If it throws an exception, however, the try-catch instead discards the exception and runs~$e_2$.
We include $\TryCatch{E}{e_2}$ as an evaluation context, so it will evaluate~$e_1$ to a value or exception, but not as a throw context so it can catch and discard an exception.

These considerations give rise to the following typing rules:
\begin{mathpar}
  \infer{ }{\Gamma \proves \Read : \sigma}
  \and
  \infer{\Gamma \proves e : \sigma}{\Gamma \proves \Write(e) : \Unit}
  \and
  \infer{ }{\Gamma \proves \Throw : \tau}
  \and
  \infer{\Gamma \proves e_1 : \tau\\ \Gamma \proves e_2 : \tau}{\Gamma \proves \TryCatch{e_1}{e_2} : \tau}
\end{mathpar}
The resulting operational semantic rules are defined on pairs of an expression and a state cell~$s$ of type~$\sigma$.
The rules for our effectful operations are as follows:
\begin{mathpar}
  \langle \Read, s \rangle \stepsone \langle s, s \rangle
  \and
  \langle \Write(v), s \rangle \stepsone \langle (), v \rangle
  \\
  \langle T[\Throw],s \rangle \stepsone \langle \Throw,s \rangle
  \\
  \langle \TryCatch{v}{e},s \rangle \stepsone \langle v,s \rangle
  \and
  \langle \TryCatch{\Throw}{e},s \rangle \stepsone \langle e,s \rangle
\end{mathpar}
We modify the rest of the operational semantics by including $s$ without modification in every other rule, as is standard for simply-typed $\lambda$-calculus with state~\citep[Chapter 13.3]{Pierce02}.

Technically, our previous noninterference theorem (Theorem~\ref{thm:dcc-ni}) still holds, and by the same proof.
However, the statement of this theorem is now very weak: it assumes that an attacker cannot see the state or distinguish \Throw from any other statement.
This allows implicit flows.
To see how exceptions allow implicit flows, consider the following example program that leaks its input to anyone who can distinguish \Throw from non-exceptional output:
$$h : \LabeledType{\ell}{\Unit + \Unit} \proves \UnlabelAs{h}{x}{\MatchSum{x}{\_}{\Throw}{\_}{\LabeledProgram{\ell}{}}} : \LabeledType{\ell}{\Unit}$$

\subsection{Ruling Out Implicit Flows}
\label{sec:ruling-out-implicit}

We now aim to eliminate implicit flows and recover a strong notion of noninterference with realistic assumptions about an attacker's power.
We achieve this result by changing our typing rules.
We associate a \emph{program-counter label} $\pc$ with the typing judgment to track the sensitivity of the context.
Thus, the typing judgment now takes the form $\Gamma \ctxsep \pc \proves e : \tau$ where $\pc$ is an information-flow label.

In the examples of implicit flows so far, vulnerabilities arose when we performed certain actions depending on the value of a secret expression.
To prevent such problems, we might update $\pc$ so that it is at least as high as any value we conditioned on in a \Match statement.
We could then ensure that actions that might leak information about those values cannot type-check in a sensitive environment.
However, there is a problem with doing this in DCC: we never \Match on labeled data.
Instead, we must first use \Unlabel, removing the label from the data before we can use that data in any way, including in a \Match expression.
This is because DCC is a \emph{coarse-grained} information-flow language.
(In fact, it is the paradigmatic coarse-grained information-flow language.)

The fact that labeled data can only be used in an \Unlabel{} expression means that the \Unlabel{} rule is the only rule in which we can reasonably increase the \pc{}.
(This may seem like a significant restriction, since we are increasing the program-counter label even when we do not \Match{} on the data we are unlabeling.
However, recent research has shown that coarse-grained information flow is equivalent to fine-grained information flow, which would increase the program-counter label in the \Match statement~\citep{RajaniG18}.)
We increase the program counter label using the join operator on labels we discussed in Section~\ref{sec:pure-ifc-lang}.
Thus, the rule for \Unlabel{} becomes the following:
$$\infer{\Gamma \ctxsep \pc \proves e_1 : \LabeledType{\ell}{\tau_1}\\ \Gamma, x : \tau_1 \ctxsep \pc \sqcup \ell \proves e_2 : \tau_2\\ \ell \prot \tau_2}{\Gamma \ctxsep \pc \proves \UnlabelAs*{e_1}{x}{e_2} : \tau_2}$$

We must also change the typing rules for functions.
Intuitively, an expression $\lambda x \ty \tau.\,e$ will execute the actions of $e$ when it is applied, not when it is defined.
It is therefore safe to \emph{construct} a $\lambda$-expression in any context, but it is only safe to \emph{apply} one in a context where its effects do not leak information.
Since we cannot, in general, know where a function will be used when it is constructed, we instead change the type of the function to restrict where it can be applied.
The type $\tau_1 \pcto \tau_2$ is a function that takes an argument of type~$\tau_1$, returns a value of type~$\tau_2$, and can be safely run in contexts which have not discriminated on anything higher than $\pc$.
This gives rise to the following two typing rules:
\begin{mathpar}
  \infer{\Gamma, x : \tau_1 \ctxsep \pc_1 \proves e : \tau_2}{\Gamma \ctxsep \pc_2 \proves \lambda x\ty\tau_1.\,e : \tau_1 \pcto[\pc_1] \tau_2} \and
  \infer{\Gamma \ctxsep \pc_1 \proves e_1 : \tau_1 \pcto[\pc_2] \tau_2\\ \Gamma \ctxsep \pc_1 \proves e_2 : \tau_1\\ \pc_1 \flowsto \pc_2}{\Gamma \ctxsep \pc_1 \proves e_1~e_2 : \tau_2}
\end{mathpar}

This change also necessitates adjusting the function protection rule from Figure~\ref{fig:dcc-protection}.
Applying a function with type $\tau_1 \pcto \tau_2$ can still reveal data through its output at the level of~$\tau_2$, but it can also reveal information about control flow up to label~$\pc$.
We therefore need to use the $\pc$ to ensure that effects will not leak information, leading to the following modified protection rule.
$$\infer{
  \ell \prot \tau_2 \\
  \ell \flowsto \pc
}{\ell \prot \tau_1 \pcto \tau_2}$$

To determine how the $\pc$ label should relate to effects, we need to be clear about what effects the attacker can and cannot see.
We assume that anyone who can see things labeled $\StateLabel$ can see the value stored in the state cell.
An attacker who can read $\StateLabel$ can see writes to state---since writes can change the stored value---but not reads---since reads leave the value unchanged.
This is a reasonable assumption in many cases, but not all.
For instance, it assumes the attacker cannot extract information through cache-based timing attacks~\citep[e.g.,][]{Kocher96}.
We also assume that any attacker who can see information labeled $\ExnLabel$ can distinguish between exceptions and other values.
An attacker who \emph{cannot} see information labeled $\ExnLabel$ is therefore not privy to the success or error status of the program, meaning they also cannot see the result if it returns successfully.
They may, however, still be able to observe the state cell if they can read~$\StateLabel$.

We can use this model to determine the \pc-based typing rules for our effectful operations.
Three of the rules are fairly simple.
The \Read rule allows any \pc, while the \Write and \Throw rules must check that the context is not too sensitive to run this computation.
\begin{mathpar}
  \infer{ }{\Gamma \ctxsep \pc \proves \Read : \sigma}
  \and
  \infer{\Gamma \ctxsep \pc \proves e : \sigma\\ \pc \flowsto \StateLabel}{\Gamma \ctxsep \pc \proves \Write(e) : \Unit}
  \and
  \infer{\pc \flowsto \ExnLabel}{\Gamma \ctxsep \pc \proves \Throw : \tau}
\end{mathpar}
The try-catch rule is slightly more complicated, since the \Catch block only executes if the try block throws an exception, and therefore may return different values depending on whether or not an exception occurs.
To ensure this control flow does not leak data, the output must be at least as sensitive as the control flow: $\ExnLabel$.
$$\infer{
  \Gamma \ctxsep \pc \proves e_1 : \tau \\
  \Gamma \ctxsep \pc \proves e_2 : \tau \\
  \ExnLabel \prot \tau
}{\Gamma \ctxsep \pc \proves \TryCatch{e_1}{e_2} : \tau}$$

None of the other rules change $\pc$, since none of the other rules can leak information about the context or change the context based on a labeled value.
These rules may appear to ensure security against the attacker we sketched above, but unfortunately this intuition misses the fact that exceptions can impact control flow outside just try-catch blocks.
Consider the following program:
\begin{equation}
  \label{eqn:exn-write-leak}
  h \ty \LabeledType{\ExnLabel}{\Unit + \Unit}, s \ty \sigma \proves \LetIn{\_}{\UnlabelAs{h}{x}{\MatchSum{x}{\_}{\Throw}{\_}{\LabeledProgram{\ExnLabel}{}}}}{\Write(s)} : \LabeledType{\ExnLabel}{\Unit}
\end{equation}
If $h = \LabeledProgram{\ExnLabel}{\Inl~()}$, this will throw an exception and the \Write will never execute.
Notably, this program leaks the value of $h$ to anyone who can see~$\StateLabel$.
\citet{PottierS02} eliminate this leak by constraining what effects can execute after an expression that may throw an exception.
While our framework can handle this generality (see Section~\ref{sec:bett-state-except}),
we take a more restrictive but far simpler approach and require that $\ExnLabel \flowsto \StateLabel$.

Proving that we have successfully eliminated data leaks requires using a notion of $\ell$-equivalence that accounts for exceptions and state.
Notably, to properly capture which attackers may view which values, our notion differs depending on how~$\ell$ relates to $\StateLabel$~and~$\ExnLabel$.

\begin{defn}
  \label{defn:state-exn-low-equiv}
  We say programs $e_1$ and $e_2$ are \emph{state and exception \mbox{$\ell$-equivalent}}, denoted $e_1 \StExnLequiv e_2$, if for all values $s$ and expression contexts $C$ such that
  $\proves s : \sigma$, ${} \ctxsep \pc \proves C[e_i] : \LabeledType{\ell}{\Unit + \Unit}$, and $\langle C[e_i], s \rangle \stepsmany \langle v_i, s_i \rangle$ for both $i = 1, 2$,
  then $v_1 = v_2$ if $\ExnLabel \flowsto \ell$ and $s_1 = s_2$ if $\StateLabel \flowsto \ell$.
\end{defn}

Intuitively, Definition~\ref{defn:state-exn-low-equiv} says that $e_1$ and $e_2$ are equivalent if no program will let an attacker at label~$\ell$ distinguish them through either the program output or the state.
Because the attacker can only see the program output if $\ExnLabel \flowsto \ell$, we only check output equivalence in that case.
Similarly, because the attacker can only see the state cell if $\StateLabel \flowsto \ell$, we only check equivalence of the state cells when the flow holds.
Note that structural equality on state values is sufficient because we assumed $\sigma$ contains no function types as subterms.

While it is possible to directly prove our type system enforces noninterference using such an equivalence relation~\citep{WayeBKCR15,RussoCH08,TsaiRH07}, we take a different approach.
We formalize the view that the \pc allows only secure effects to simplify the noninterference proof and help avoid the need for clever ad-hoc reasoning, such as the argument we made for try-catch.

\subsection{Tracking Effects}
\label{sec:tracking-effects-state-exns}

To show that the $\pc$ label restricts well-typed programs to be those with secure effects, we need to track the effects in a program.
To do so, we use a standard type-and-effect system~\citep{Lucassen88,Nielson96,NielsonN99,MarinoM09}.
This type-and-effect system assigns each step of a typing proof an effect $\e$ from a set \Eff of possible effects.
Therefore, we need to decide on the contents of \Eff.

So far, we have described our language has having two effects: state and exceptions.
We might therefore let $\Eff = \{\StateEff, \ExnEff\}$ where $\StateEff$ represents state and $\ExnEff$ exceptions.

This choice is undesirable for two reasons.
First, consider the following program (where the state is of type $\Unit + \Unit$):
$$x : \Unit + \Unit \proves \MatchSum{x}{\_}{\Read}{\_}{\Throw} : \Unit + \Unit$$
This program could either read state or throw an exception.
So which effect should we give it?
We must note that both are possible, which we can do by having an effect that represents ``can use state and/or throw an exception.''
In fact, we need an effect for each possible \emph{collection} of our effects.
Thus our possible effects come from the power set, so $\Eff = 2^{\{\StateEff, \ExnEff\}}$.
Notably, this gives our effects a nice lattice structure, as is standard for a power set.

Second, considering state as a single effect does not allow us to represent that our attacker can only see writes.
For instance, consider the following program in the same setting as above:
$$h : \LabeledType{\ell}{\Unit + \Unit} \proves \UnlabelAs{h}{x}{\MatchSum{x}{\_}{\LabeledProgram{\ell}{\Read}}{\_}{\LabeledProgram{\ell}{\Inl~()}}} : \LabeledType{\ell}{\Unit + \Unit}$$
An attacker who cannot distinguish values labeled $\ell$ cannot distinguish which branch the program takes.
However, if we were to keep state as a single effect, we would have to label that branch as having the state effect, and therefore disallow it.
To resolve this problem, we separate state operations into read effects \REff and write effects \WEff and change \Eff to $2^{\{\REff, \WEff, \ExnEff\}}$.
Note that since our attacker cannot see reads at all, we could consider \Read to be a pure operation.
Though this would simplify a few technical details, we find it is more intuitive to include \REff as an effect.

Now we can develop our type-and-effect system.
We again change the form of the typing judgment, this time to $\Gamma \proves e : \tau \ctxsep \varepsilon$, where $\varepsilon \subseteq \{\REff, \WEff, \ExnEff\}$.
For readability, we will often write this set without curly brackets.
The example program above that could read state or throw an exception would therefore have the typing judgment $\Gamma \proves e : \LabeledType{\ell}{\Unit + \Unit} \ctxsep \REff, \ExnEff$.

Since we are trying to analyze the security of the program, we create a function $\ell_{-}$ from effects to labels, associating a label~$\ell_\e$ with each effect~$\e$.
This label corresponds to our attacker model of who can observe the effect.
For $\WEff$ and $\ExnEff$ we already have these labels---$\StateLabel$ and $\ExnLabel$, respectively.
Because reads are not visible, we can set $\ell_\REff = \top$.
For other effects, the label $\ell_\e$ should capture who might observe \emph{any} component of $\e$, so it should be a lower bound on the components of~$\e$.
That is, $\ell_{\WEff,\ExnEff} \flowsto \StateLabel$ and $\ell_{\WEff,\ExnEff} \flowsto \ExnLabel$.

Now we can build our type-and-effect system, modifying the rules of the pure typing system.
As in the $\pc$ case, most of the typing rules do not change $\e$, since most terms do not change the effects a program may run.
For the same reasons as before, functions and the four rules we added explicitly for effects do change.
Since $\lambda$-expressions execute effects when applied but not when defined, we take the same approach as before and record the effects in the type.
We also modify the function protection rule as we did previously.
This gives rise to the following rules:
\begin{mathpar}
  \infer{\Gamma, x \ty \tau_1 \proves e : \tau_2 \ctxsep \varepsilon}{\Gamma \proves \lambda x \ty \tau_1.\,e : \tau_1 \effto \tau_2 \ctxsep \varnothing}
  \and
  \infer{\Gamma \proves e_1 : \tau_1 \effto[\varepsilon_1] \tau_2 \ctxsep \varepsilon_2\\\\ \Gamma \proves e_2 : \tau_1 \ctxsep \varepsilon_3}{\Gamma \proves e_1~e_2 : \tau_2 \ctxsep \varepsilon_1 \cup \varepsilon_2 \cup \varepsilon_3}
  \and
  \infer{
    \ell \prot \tau_2 \\
    \ell \flowsto \ell_\e
  }{\ell \prot \tau_1 \effto \tau_2}
\end{mathpar}
The rules for our extended expressions must incur appropriate effects.
Again, three are fairly simple.
\begin{mathpar}
  \infer{ }{\Gamma \proves \Read : \sigma \ctxsep \REff}
  \and
  \infer{\Gamma \proves e : \sigma \ctxsep \varepsilon}{\Gamma \proves \Write(e) : \Unit \ctxsep \varepsilon \cup \WEff}
  \and
  \infer{ }{\Gamma \proves \Throw : \tau \ctxsep \ExnEff}
\end{mathpar}
Reading a value gives an \REff~effect, while $\Write(e)$ evaluates $e$ and therefore any effects e runs, and then runs a write effect.
Throwing an exception creates an \ExnEff~effect, but as before, catching an exception is more complicated.
A try-catch expression does not generate any new effects, though it can combine effects from both the \Try and \Catch blocks.
More importantly, since we are still aiming to enforce security through our type-and-effect system, the security concerns surrounding control flow still apply.
We therefore again require $\ExnLabel \prot \tau$.
The resulting rule is as follows.
$$\infer{
  \Gamma \proves e_1 : \tau \ctxsep \e_1 \cup \ExnEff \\
  \Gamma \proves e_2 : \tau \ctxsep \e_2 \\
  \ExnLabel \prot \tau \\
}{\Gamma \proves \TryCatch{e_1}{e_2} : \tau \ctxsep \e_1 \cup \e_2}$$

We also modify the $\Unlabel$ rule to prevent effects from leaking data.
Specifically, if a program $e_2$ relies on data with label $\ell$, the effects of $e_2$ can only be visible at or above~$\ell$.
Including the existing output restriction from Section~\ref{sec:pure-ifc-lang} gives us the following rule:
$$\infer{
  \Gamma \proves e_1 : \LabeledType{\ell}{\tau_1} \ctxsep \e_1 \\
  \Gamma, x\ty\tau_1 \proves e_1 : \tau_2 \ctxsep \e_2 \\
  \ell \prot \tau_2 \\
  \ell \flowsto \ell_{\e_2}
}{\Gamma \proves \UnlabelAs*{e_1}{x}{e_2} : \tau_2 \ctxsep \e_1 \cup \e_2}$$
Finally, as several of these rules require precisely equal effects, we include a rule allowing judgments to overstate a program's effect:
$$\infer{
  \Gamma \proves e : \tau \ctxsep \e' \\
  \e' \subseteq \e
}{\Gamma \proves e : \tau \ctxsep \e}$$

We now have the mechanics we need to formally state a connection between the program-counter label and effects.
Intuitively, a program with a sufficiently restrictive \pc cannot have certain effects.
We capture this intuition with the following lemma:
\begin{lem}[Connection between Program-Counter Label and Effects]
  \label{lem:pc-effects-state-exn}
  The program-counter label forces effects to be well typed, so if $\Gamma \ctxsep \pc \proves e : \tau$, then $\Gamma \proves e : \tau \ctxsep \e$ for some~$\e$.
  Moreover, the $\pc$ controls which effects are possible.
  In particular:
  \begin{itemize}
    \item If $\pc \nflowsto \StateLabel$ and $\Gamma \ctxsep \pc \proves e : \tau$, then there exists an $\e$ such that $\Gamma \proves e : \tau \ctxsep \e$ and $\WEff \notin \e$.
    \item If $\pc \nflowsto \ExnLabel$ and $\Gamma \ctxsep \pc \proves e : \tau$, then there exists an $\e$ such that $\Gamma \proves e : \tau \ctxsep \e$ and $\ExnEff \notin \e$.
  \end{itemize}
\end{lem}
Since our type-and-effect system restricts to secure effects, the first part of Lemma~\ref{lem:pc-effects-state-exn} is non-trivial.
Also note that the two type systems use slightly different type constructors for functions---the \pc system includes labels while the type-and-effect system includes effects.
We implicitly convert between the two here, as the label associated with each effect makes the conversion simple.
See Appendix~\ref{sec:full-type-transl} for the full details of the conversion.

Lemma~\ref{lem:pc-effects-state-exn} provides a direct connection between program-counter labels and effects.
From the point-of-view we have been advocating---that program-counter labels limit programs to secure effects---this is the semantics of a program-counter label.
We will discuss this semantics for program-counter labels in more depth in Section~\ref{sec:prod-effects-ifc}.

\subsection{Effectful Noninterference Half-Off}
\label{sec:effectf-nonint-free}

We now aim to use our effect tracking to prove noninterference for our extended language.
We use a strategy from work on the semantics of type-and-effect systems and give meaning to effects via translation into pure programs.
Importantly, if we do this in such a way that our notions of ``low-equivalent'' match up, we can get noninterference automatically.

Our translation uses monads to represent effects, as is common.
In fact, we have one monad per (set of) effect(s), defined as follows.
For technical reasons, we use the same monad for $\{\WEff\}$ and $\{\REff,\WEff\}$, as well as for $\{\WEff,\ExnEff\}$ and $\{\REff,\WEff,\ExnEff\}$.
\footnote{This is because the \emph{writer} monad, for the write effect without read, is not a monad unless $\sigma$ is a monoid.
  However, the \emph{state} monad, for read and write effects, is a cartesian monad no matter the type of $\sigma$.}
We also assume for simplicity that $\sigma$, the type of our state cell, already has at least $\StateLabel$ sensitivity---%
that is, $\StateLabel \prot \sigma$---allowing us to omit explicit use of~$\StateLabel$.
If this were not the case, we would use $\LabeledType{\StateLabel}{\sigma}$ instead of~$\sigma$.
\begin{center}
\begin{tabular}{c|l}
  $\e$ & \multicolumn{1}{c}{$\Monad{\varepsilon}(\tau)$}\\
  \hline
  $\varnothing$ & $\tau$\\
  $\{\REff\}$ & $\sigma \to \tau$\\
  $\{\ExnEff\}$ & $\LabeledType{\ExnLabel}{\Unit + \tau}$\\
  $\{\WEff\}$ and $\{\REff,\WEff\}$ & $\sigma \to (\tau \times \sigma)$\\
  $\{\REff,\ExnEff\}$ & $\sigma \to \LabeledType{\ExnLabel}{\Unit + \tau}$\\
  $\{\WEff,\ExnEff\}$ and $\{\REff,\WEff,\ExnEff\}$ & $\sigma \to (\LabeledType{\ExnLabel}{\Unit + \tau} \times \sigma)$
\end{tabular}
\end{center}
We refer to the monad for a (set of) effect(s)~$\varepsilon$ as $\Monad{\varepsilon}$.
We use this notation instead of the more-traditional $M_{-}$ as we will generalize to a productor in Section~\ref{sec:prod-effects-ifc}.
These monads are standard, but they are not automatic.
Instead, they reflect some choices about the semantics of programs.
For instance, the fact that the monad for the set $\{\REff, \WEff{},\ExnEff{}\}$ is $\sigma \to (\Unit + \tau) \times \sigma$ instead of $\Unit + (\sigma \to \tau \times \sigma)$ reflects the fact that state persists even when an exception is thrown.

Note that we can define three special kinds of programs:
\begin{itemize}
\item For any set $\varepsilon$ and program $\Gamma, x \ty \tau_1 \proves p : \Monad{\varepsilon}(\tau_2)$, we can define $\Gamma, x \ty \Monad{\varepsilon}(\tau_1) \proves \BindM{\varepsilon}(p) : \Monad{\varepsilon}(\tau_2)$.
\item For any type $\tau$ and set $\varepsilon$, we can define $\eta_\varepsilon : \tau \to \Monad{\varepsilon}(\tau)$.
\item For any type $\tau$ and pair of sets $\varepsilon_1$ and $\varepsilon_2$ such that $\varepsilon_1 \subseteq \varepsilon_2$, we can define a program\linebreak \mbox{$\CoerceE{\varepsilon_1}{\varepsilon_2} : \Monad{\varepsilon_1}(\tau) \to \Monad{\varepsilon_2}(\tau)$}.
\end{itemize}
This (along with some easily-proven properties of these programs) makes $\Monad{-}$ an \emph{indexed monad}~\citep{WadlerT98,OrchardPM14}, which is a mathematical object that gives semantics to systems of effects.
Note that this requires $L_{-}$ itself to be an indexed monad, which was proven by \citet{AbadiBJHR99}.
As the name suggests, indexed monads are a generalization of monads.
Indexed monads also give a standard way to translate effectful programs into pure programs, transforming a derivation of $\Gamma \proves e : \tau \ctxsep \varepsilon$ into a program~$e'$ such that $\Gamma \proves e' : \Monad{\varepsilon}(\tau)$.
See Appendix~\ref{sec:full-transl-example} for the full translation.

This translation generates an important security result.
Because all well-typed pure programs guarantee noninterference, we can extend that result to our effectful language if our translation has two specific properties.
First, it must be sound.
Indeed, if we omit any of the careful reasoning about exceptions from Section~\ref{sec:tracking-effects-state-exns} our translation would be unsound.
The point at which the soundness proof breaks down is, however, often informative.
For example, when translating Program~\ref{eqn:exn-write-leak},
our translation would need to return a value of type $\LabeledType{\ExnLabel}{\Unit + \Unit} \times \sigma$ after unlabeling a value of type $\LabeledType{\ExnLabel}{\Unit + \Unit}$.
Satisfying the premise of the \Unlabel rule that the removed label protects the output type forces exactly the $\ExnLabel \flowsto \StateLabel$ assumption we made above.

The second requirement to obtain effectful noninterference is that our monadic translation faithfully translates our effectful notion of equivalence to our pure one.
In this case it does because monadic programs simulate effectful programs.
Without labels, this is a well-known theorem of \citet{WadlerT98}, adding labels does not significantly change the proof. 
We can therefore use our type-and-effect system and monadic translation to make a strong claim of noninterference for the \pc system.
\begin{thm}[Noninterference for State and Exceptions]
  \label{thm:state-exn-ni}
  For all expressions $e_1$ and $e_2$, and for all $\ell \in \Labs$, if $\Gamma \ctxsep \pc \proves e_1 : \tau$ and $\Gamma \ctxsep \pc \proves e_2 : \tau$,
  and $\ell \prot \tau$ and $\ell \flowsto \pc$, then for all labels $\AtkLabel \in \Labs$, either $\ell \flowsto \AtkLabel$ or $e_1 \StExnLequiv[\AtkLabel] e_2$.
\end{thm}

\begin{proof}
  This is a special case of the Noninterference Half-Off Theorem (Theorem~\ref{thm:niho}) which uses the fact that effectful programs are equivalent if their monadic translations are equivalent.
  Theorem~\ref{thm:niho} also relies on the fact that pure programs are noninterfering (Theorem~\ref{thm:dcc-ni}).
\end{proof}

The requirement that $\ell \flowsto \pc$ may appear odd next to classic definitions of noninterference.
We require such a flow because our contextual notion of equivalence treats $e_1$ and $e_2$ as program inputs, but they may have effects.
This requirement constrains those effects so that if $\ell \nflowsto \AtkLabel$, then $\AtkLabel$ will be unable to see them.
In most classic noninterference statements, the inputs are values, meaning this restriction is unnecessary as any well-typed value type-checks with $\pc = \top$.

\section{Example: Termination-Sensitive Noninterference}
\label{sec:tsni}

Type-and-effect systems can tell us more about programs than whether they access state or throw exceptions.
One classic application is checking termination by considering possible nontermination to be an effect.
We show that we can treat possible nontermination as a \emph{secure} effect.
This explains the role of program-counter labels in ruling out termination leaks.
That is, we show how \emph{termination-sensitive noninterference} falls out of our framework.

The fragment of DCC we use in Section~\ref{sec:pure-ifc-lang} is strongly-normalizing; thus, all programs terminate.
Even the extensions in Section~\ref{sec:state-and-exns} did not allow for nonterminating behavior, since we require that the type of the state cell is first order.
We can, however, easily add a standard fixpoint operator:
\begin{syntax}
  \category[{Expressions}]{e} \alternative{\cdots} \alternative{\Fixpoint{f}{\tau}{e}}
\end{syntax}
\begin{mathpar}
  \infer{\Gamma, f \ty \tau \proves e : \tau}{\Gamma \proves \Fixpoint{f}{\tau}{e} : \tau}
  \and
  \Fixpoint{f}{\tau}{e} \stepsone \subst{e}{f}{\Fixpoint{f}{\tau}{e}}
\end{mathpar}

Because programs may not terminate in this extended language, the fact that our definition of $\ell$-equivalence (Definition~\ref{defn:ell-equiv}) allows for different termination behavior matters.
In particular, our previous noninterference theorem (Theorem~\ref{thm:dcc-ni}) now says that if both programs terminate, they must produce the same value.
If \emph{either} program diverges, however, it makes no guarantees.

This guarantee models an attacker who cannot tell if a program has failed to terminate, or if it will produce an output on the next step.
That is, the attacker is \emph{insensitive} to termination behavior.
This notion of noninterference is therefore called \emph{termination-insensitive} noninterference.

While DCC with the call-by-value semantics we are using enforces termination-insensitive noninterference~\citep{AbadiBJHR99,HeintzeR98,BowmanA15,AlgehedB19},
termination channels can leak arbitrary amounts of data~\citep{AskarovHSS08}.
We would therefore like to remove the strong assumption that attackers cannot use those channels.
To do so, we define a stronger form of equivalence, \emph{termination-sensitive $\ell$-equivalence}, and use it to analyze security.
\begin{defn}[Termination-Sensitive $\ell$-Equivalence]
  \label{defn:ts-ell-equiv}
  We say $e_1$ is \emph{termination-sensitive $\ell$-equivalent} to $e_2$, denoted \mbox{$e_1 \tslequiv e_2$}, if for all expression contexts $C$ such that
  $\proves C[e_i] : \LabeledType{\ell}{\Unit + \Unit}$ for both $i = 1, 2$, then $C[e_1] \stepsmany v$ if and only if $C[e_2] \stepsmany v$.
\end{defn}

We follow the same approach as in Section~\ref{sec:state-and-exns} to ensure noninterference with respect to this stronger definition:
we restrict effects with a \pc~label, build a corresponding type-and-effect system, and prove security by translating to DCC with its termination-insensitive guarantee.
Because nontermination from $\Fix$ is the only effect, this process is considerably simpler than in Section~\ref{sec:state-and-exns}.

Traditionally, termination-sensitive noninterference assumes all attackers can see whether or not a program terminates.
We take a more general approach and assume that some attackers are termination-sensitive, while others may not be.
Specifically, we imagine there is a label $\TermLabel \in \Labs$ such that any attacker who can read $\TermLabel$ will eventually infer information from nontermination, but others will not.
This gives rise to the following rule:
$$\infer{\Gamma, f \ty \tau \ctxsep \pc \proves e : \tau\\ \pc \flowsto \TermLabel}{\Gamma \ctxsep \pc \proves \Fixpoint{f}{\tau}{e} : \tau}$$
The rule ensures that only data available at label $\TermLabel$---visible to any termination-sensitive attacker---can influence the program's termination behavior.
This single label-based rule allows us to model termination-insensitivity by setting $\TermLabel = \top$,
model traditional termination-sensitivity using a bottom label~$\bot$ by setting $\TermLabel = \bot$,
or express policies about other levels of termination visibility.

We can now move on to the type-and-effect system.
This time it has only two possible effects: $\varnothing$~or $\PntEff$ with labels $\top$ and $\TermLabel$, respectively.
The typing rule for the fixed-point operator is then:
$$\infer{\Gamma, f \ty \tau \proves e : \tau \ctxsep \varepsilon}{\Gamma \proves \Fixpoint{f}{\tau}{e} : \tau \ctxsep \PntEff}$$
The $\Unlabel$ rule constrains effects in the same way as in Section~\ref{sec:state-and-exns}, and no other typing rules change or constrain effects.
We also change function types as in Section~\ref{sec:state-and-exns}.

Finally, we translate this type-and-effect system into pure DCC using a monadic translation.
Unfortunately, this time there is no such monad definable in the language of Section~\ref{sec:pure-ifc-lang}.
Luckily, the original definition of DCC~\citep{AbadiBJHR99} allowed for nontermination with a fixed-point operator, but only if the result was a \emph{pointed} type.
We thus add fixed points and pointed types---the parts of DCC we omitted from Section~\ref{sec:pure-ifc-lang}---and use pointed types to define our monad.
Pointed types are also used in defining the denotational semantics of DCC.
\citet{AbadiBJHR99} define the denotational semantics of DCC using Scott~domains, and give pointed types semantics using domains with a bottom element.
Then, fixed points of computations in a pointed types can be found as usual.
\begin{syntax}
  \category[{Types}]{\tau} \alternative{\cdots} \alternative{\LiftType{\tau}}
  \category[{Expressions}]{e} \alternative{\cdots} \alternative{\LiftProgram{e}} \alternative{\SeqTerm*{x}{e_1}{e_2}} \alternative{\Fixpoint{f}{\tau}{e}}
  \category[Evaluation Contexts]{E} \alternative{\cdots} \alternative{\LiftProgram{E}} \alternative{\SeqTerm*{x}{E}{e}}
\end{syntax}
The type~$\LiftType{\tau}$ represents a version of~$\tau$ that supports fixed points, while the expression~$\Lift(e)$ lifts the expression~$e$ from type $\tau$~to~$\LiftType{\tau}$.
The expression $\SeqTerm*{x}{e_1}{e_2}$ waits for $e_1$ to terminate and, if it does, binds the result to $x$ in $e_2$.
Finally, the term $\Fixpoint{f}{\tau}{e}$ defines a fixpoint.

We then define a judgment determining when a type is a pointed type as follows:
\begin{mathpar}
  \infer{ }{\proves \LiftType{\tau}~\Ptd} \and
  \infer{\proves \tau_1~\Ptd\\ \proves \tau_2~\Ptd}{\proves \tau_1 \times \tau_2~\Ptd}
  \and
  \infer{\proves \tau~\Ptd}{\proves \LabeledType{\ell}{\tau}~\Ptd} \and
  \infer{\proves \tau_2~\Ptd}{\proves \tau_1 \to \tau_2~\Ptd}
\end{mathpar}
This allows us to state the following typing and semantic rules for the newly-added terms:
\begin{mathpar}
  \infer{\Gamma \proves e : \tau}{\Gamma \proves \LiftProgram{e} : \LiftType{\tau}} \and
  \infer{\Gamma \proves e_1 : \LiftType{\tau_1}\\\\\Gamma, x : \tau_1 \proves e_2 : \tau_2\\ \proves \tau_2~\Ptd}{\Gamma \proves \SeqTerm*{x}{e_1}{e_2} : \tau_2} \and
  \infer{\Gamma, f \ty \tau \proves e : \tau\\ \proves \tau~\Ptd}{\Gamma \proves \Fixpoint{f}{\tau}{e} : \tau}\\
  \SeqTerm*{x}{\LiftProgram{v}}{e} \stepsone \subst{e}{x}{v} \and
\end{mathpar}

Because the possibility of nontermination is limited to pointed types, we consider the full DCC to be pure for our purposes.
That is, we consider programs which are nonterminating, but which have a pointed type, pure.
To see why this is justified, we have to ask what we consider an effect.
We can probe this by considering the example from Section~\ref{sec:state-and-exns}: why do we consider a program which returns a value of type~$\sigma \to (\tau \times \sigma)$ to be pure, but not a program which accesses state and returns a value of type~$\tau$?
After all, they can encode the same computations.
Intuitively, though, we have translated accesses to state to operations provided by the more-complex type; namely, reads have been replaced by usage of the parameter of type~$\sigma$, and writes by returning an appropriate result of type~$\sigma$.
In the current case, we have translated fixpoint computations which can take place anywhere with a fixpoint operation provided by a more-complex type.
This eases reasoning in many settings.
For instance, when attempting a proof by logical relations, all reasoning about nontermination can now be located in pointed types.

Some readers may remain skeptical of the application of the word ``pure'' to DCC extended with pointed types.
It is therefore worth noting that to use noninterference half-off, we only need a monad which can represent the effect that we want to reason about in a language where it is easier to prove noninterference.
Proving noninterference in full (call-by-value) DCC with pointed types is easier than in our \pc system for two reasons.
First, possible nontermination is located by types, which makes many proof techniques easier, as noted above.
Second, call-by-value DCC with pointed types enforces \emph{termination-insensitive} noninterference, which is generally simpler to prove than termination-sensitive noninterference, which our \pc system enforces.

While pointed types give access to fixpoint operators, the translation of the \Unlabel rule fails whenever a program returns a $\LiftType{\tau}$, because $\LiftType{\tau}$ protects no labels.
This was \citeposessive{AbadiBJHR99} original design, to ensure that labeled data can never determine the termination behavior of a program.
This is too restrictive for us, since we want to allow data up to label $\TermLabel$ to influence a program's termination behavior.
Of course, the label must also be allowed to influence any output the program produces if it does terminate.
This leads to the following rule:
$$\infer{\ell \flowsto \TermLabel\\ \ell \prot \tau}{\ell \prot \LiftType{\tau}}$$
Note that, even when setting $\TermLabel = \bot$, this rule allows \emph{public} data to influence termination behavior.
We therefore differ slightly from \citeposessive{AbadiBJHR99} definition in systems that distinguish public data from unlabeled data.

With this rule, the traditional monadic translation works, which tells us that the type-and-effect system indeed enforces noninterference.
Now we need to connect the \pc system to the type-and-effect system in an analogous way to Lemma~\ref{lem:pc-effects-state-exn}.
If the \pc cannot influence $\TermLabel$, then the program \emph{must} terminate.
We can formalize this as follows:
\footnote{As with Lemma~\ref{lem:pc-effects-state-exn}, we implicitly assume a simple translation between the two sets of type constructors (see Appendix~\ref{sec:full-type-transl}).}
\begin{lem}
  \label{lem:pc-tsni}
  If $\Gamma \ctxsep \pc \proves e : \tau$, then $\Gamma \proves e : \tau \ctxsep \PntEff$.
  Moreover, if $\pc \nflowsto \TermLabel$, then $\Gamma \proves e : \tau \ctxsep \varnothing$.
\end{lem}
\noindent
The first statment is again non-trivial since our type-and-effect system restricts to secure effects.

Lemma~\ref{lem:pc-tsni} is a very powerful guarantee.
Combined with a monadic translation based on pointed types (which can be found in Appendix~\ref{sec:full-transl-example}), it tells us that if the $\pc$ is too high, $e$ will terminate.
Thus, if some data determines whether or not a program terminates, it must be visible to any attacker who can see termination behavior.
Formalizing this, we get the following guarantee:
\begin{thm}[Termination-Sensitive Noninterference]
  For $\ell \in \Labs$ and expressions $e_1$ and $e_2$, if $\Gamma \ctxsep \pc \proves e_1 : \tau$ and $\Gamma \ctxsep \pc \proves e_2 : \tau$,
  and $\ell \prot \tau$ and $\ell \flowsto \pc$, then for all labels $\AtkLabel \in \Labs$ where $\TermLabel \flowsto \AtkLabel$, either $\ell \flowsto \AtkLabel$ or $e_1 \tslequiv[\AtkLabel] e_2$.
\end{thm}
This is a special case of Theorem~\ref{thm:niho}.
As with both previous noninterference theorems (Theorems~\ref{thm:dcc-ni} and~\ref{thm:state-exn-ni}), this theorem says data at label~$\ell$ cannot leak to an attacker who cannot distinguish values at level~$\ell$.
This time, however, the attacker can glean information from nontermination.
Moreover, if we set $\TermLabel = \bot$, then $\bot \flowsto \AtkLabel$ for all $\AtkLabel$, so if $\ell \neq \bot$ we get classic termination-sensitive noninterference.

Note that if either $e_1$ or $e_2$ may diverge, the theorem always allows a termination-sensitive attacker to distinguish them.
In this case, Lemma~\ref{lem:pc-tsni} ensures $\pc \flowsto \TermLabel$, which then guarantees $\ell \flowsto \AtkLabel$ by transitivity.

\section{A Framework for Effectful Labeled Languages}
\label{sec:prod-effects-ifc}

We have now twice given semantics to \pc systems using type-and-effect systems and monadic translations.
The ability to give semantics to not just traditional effects like state, but combinations of effects and more unusual effects, like nontermination, demonstrates the power of this technique.
We now generalize these ideas by moving to a semantic framework that does not lock us into a single language.
Instead, we provide a set of typing rules and equations specifying the language features that our semantics require.

This approach allows us to describe the semantics of a large class of languages at once.
By making our framework as general as possible, we learn about what features a language needs to make our semantics work.
More importantly, we can also be sure that we do not rely on a \emph{lack} of other language features.
Thus, we can make strong semantic and security guarantees about any language that admits the rules of our framework, regardless of what other features may be present.
We develop our framework by looking at the commonalities in our examples and determining which properties are necessary to obtain the security results.

In our examples in Sections~\ref{sec:state-and-exns} and \ref{sec:tsni}, we developed a semantics for \pc labels via two layers of translation.
We first translated from a \pc system to a type-and-effect system, and then to a pure language via monadic translation.
Our general framework makes the same division.
We begin by focusing on the monadic translation and discuss the first layer in Section~\ref{sec:niho}.

The monadic translation required two languages, one effectful and one pure, a set of effects~$\Eff$, and a translation capturing effectful programs as monadic pure ones.
The effectful language used a type-and-effect system with judgements of the form $\Gamma \proves e : \tau \ctxsep \e$ with $\e \in \Eff$, while the pure language had judgements of the form $\Gamma \proves e : \tau$.
We used an indexed monad to provide a type transformer $\Monad{\e}(-)$ for each effect $\e \in \Eff$, so that $\Monad{\e}(\tau)$ was the pure type resulting from from translating an effectful program with type $\tau$ and effect $\e$.

Our framework generalizes this approach.
We again have two languages, one effectful and one pure.
Here the type-and-effect judgements in the effectful language take the form $\effColor{t \provesEff p \dashv t' \ctxsep \e}$ and the typing judgements of the pure language take the form $\pureColor{\tau \provesPure \rho \dashv \tau'}$.
For clarity, we annotate the turnstyles and color judgments in the two languages differently.
We also use Roman letters to refer to types and programs in the effectful language and Greek letters in the pure language.

An eagle-eyed reader will have noticed that our examples have a context on the left, while our framework has only a single~$\effColor{t}$ or~$\pureColor{\tau}$.
We make this choice due to the categorical nature of monadic semantics, which makes it easiest to talk about \emph{single-input, single-output} systems.
This structure might appear restrictive, but actually allows a great deal of generality.
A common trick, which we use here, is to have $\effColor{t}$~and~$\pureColor{\tau}$ represent \emph{contexts}, rather than types.
However, the multiple-input, single-output nature of our example languages actually gives more complex structure.
That structure allows us to interpret the categorical operations as applying to a single input of programs, which we do freely in our examples.
We formalize the technical details of this transformation in Appendix~\ref{sec:multi-input}.

In both of our examples, our set of effects $\Eff$ were the power set of the individual effects in the language.
The result was a lattice structure that we leveraged to define the effect of sequentially composing effectful programs.
Sequential composition took the form of function application.
If we wrap $e_2$ in a lambda and provide $e_1$ as an argument, we first execute $e_1$ and then $e_2$.
In particular, the abstraction, application, and effect variance rules in Section~\ref{sec:tracking-effects-state-exns}
combined prove that the following rule is admissible.
$$\infer{
  \Gamma \proves e_1 : \tau_1 \ctxsep \e_1 \\
  \Gamma, x\ty\tau_1 \proves e_2 : \tau_2 \ctxsep \e_2 \\
  \e_1 \cup \e_2 \subseteq \e
}{\Gamma \proves (\lambda x.\, e_2)~e_1 : \tau_2 \ctxsep \e}$$
For our general framework, we also require both a set of effects~$\Eff$ and a sequential composition operation that we denote $p_1 \seq p_2$.
Note that, as in our examples, the language need not have an explicit sequential composition operation;
one merely needs to be macro-expressible~\mbox{\citep{Felleisen1990}}.
In Sections~\ref{sec:state-and-exns} and~\ref{sec:tsni}, the lattice structure of $\Eff$ allowed us to compose any pair of effectful programs.
In general, however, this requirement is not only unnecessary, it is overly restrictive.
It is sometimes useful for the composition operation to be \emph{partial}, allowing only certain sequences of effects---and certain sequences of effectful programs---to compose.
In particular, we discuss in Section~\ref{sec:bett-state-except} how using partiality, we can lift a seemingly-arbitrary restriction in our treatment of exceptions and state.
For our framework, we therefore turn to \citeposessive{Tate13} \emph{effector}, which was designed as the minimal structure required to give meaning to such compositions.

An effector is a set \Eff with a relation $\ComposeE{-, \dotsc, -}{-}$ defining how the effects can compose.
Intuitively, $\ComposeE{\e_1, \dotsc, \e_n}{\e}$ means that sequentially composing $n$ programs with effects $\e_1$ through $\e_n$, respectively, can result in a program with effect $\e$.
Note that $n$ may in particular be zero or one, where $\ComposeE{}{\e}$ means that a program judged to have effect $\e$ may be pure, and $\ComposeE{\e}{\e'}$ means that a program judged to have effect $\e'$ may also have effect $\e$.
The relation must follow appropriate versions of identity and associativity laws, reflecting these intuitions~\citep[Section 5]{Tate13}.
For the power-set lattices from Sections~\ref{sec:state-and-exns} and~\ref{sec:tsni}, we can simply define $\ComposeE{\e_1, \dotsc, \e_n}{\e}$ as $\e_1 \cup \dotsb \cup \e_n \subseteq \e$.

\begin{figure*}
  \small
  \begin{mathpar}
    \infer*[left=Seq$_\leq$]{
      \effColor{t_0 \provesEff p_1 \dashv t_1 \ctxsep \e_1} \\
      \cdots \\
      \effColor{t_{n-1} \provesEff p_n \dashv t_n \ctxsep \e_n} \\\\
      \ComposeE{\e_1, \dotsc, \e_n}{\e}
    }{\effColor{t_0 \provesEff p_1 \seq \dotsb \seq p_n \dashv t_n \ctxsep \e}}
    \and
    \infer*[left=Seq]{
      \pureColor{\tau_0 \provesPure \rho_1 \dashv \tau_1} \\
      \cdots \\
      \pureColor{\tau_{n-1} \provesPure \rho_n \dashv \tau_n}
    }{\pureColor{\tau_0 \provesPure \rho_1 \seq \dotsb \seq \rho_n \dashv \tau_n}}
    \\
    \infer*[left=Capture]{\effColor{t \provesEff p \dashv t' \ctxsep \e}}{\pureColor{\pureTypeTrans{t} \provesPure \Captured{\e}{p} \dashv \CapturedPureTypeTrans{\e}{t'}}}
    \\
    \infer*[left=Map]{\pureColor{\tau \provesPure \rho \dashv \tau'}}{\pureColor{\CapturedType{\varepsilon}{\tau} \provesPure \MapE{\varepsilon}(\rho) \dashv \CapturedType{\varepsilon}{\tau'}}}
    \and
    \infer*[left=Join]{\ComposeE{\varepsilon_1, \ldots, \varepsilon_n}{\varepsilon}}{\pureColor{\CapturedType{\varepsilon_1}{\cdots\CapturedType{\varepsilon_n}{\tau}\cdots} \provesPure \JoinE{\varepsilon_1, \ldots, \varepsilon_n}{\varepsilon} \dashv \CapturedType{\varepsilon}{\tau}}}
    \\
    \infer*[left=CapturedSeq]{
      \effColor{t_0 \provesEff p_1 \dashv t_1 \ctxsep \e_1} \\
      \cdots \\
      \effColor{t_{n-1} \provesEff p_n \dashv t_n \ctxsep \e_n} \\
      \ComposeE{\e_1, \dotsc, \e_n}{\e}
    }{\pureColor{\Captured{\e}{p_1 \seq \dotsb \seq p_n}
      \progEq{\pureTypeTrans{t_0}, \CapturedPureTypeTrans{\e}{t_n}}
      \Captured{\e_1}{p_1} \seq \MapE{\e_1}(\Captured{\e_2}{p_2} \seq \MapE{\e_2}(\dotsb (\Captured{\e_n}{p_n}))) \seq \JoinE{\e_1, \dotsb, \e_n}{\e}}}
    \\
    \infer*[left=ProtectC]{
      \ell \flowsto \ell_\e \\
      \effColor{\ell \prot t}
    }{\pureColor{\ell \prot \CapturedPureTypeTrans{\e}{t}}}
    \and
    \infer*[left=EquivCap]{
      \pureColor{\Captured{\e}{p} \lequiv \Captured{\e}{q}}
    }{\effColor{p \lequivEff{\e} q}}
  \end{mathpar}
  \caption[Main Rules for our Semantic Framework]{Main Rules for Semantic Framework of Labeled Pure and Effectful Programming Languages}
  \label{fig:framework-rules}
\end{figure*}

An effector allows us to state that composition of effectful programs is only required when their types match and their effects compose.
We formalize this rule as \textsc{Seq$_\leq$} in Figure~\ref{fig:framework-rules},
which also requires the larger program's effects to be a valid composition of the individual effects.
Note that, as with composing effects, \textsc{Seq$_\leq$} requires composition of zero or more programs.
The above typing rule only demonstrates composition in our example languages for pairs of programs.
Using pairwise composition, we can inductively define composition of any larger number of programs.
Composing a single program is just that program unmodified, and nullary composition is the identity~$\lambda x.\, x$.

Our pure language, DCC, also used function application for sequential composition.
Our general framework requires sequential composition for the pure language in the \textsc{Seq} rule.
Because the language is pure, there is no concern about when effects may compose, so we
require pure programs to compose whenever the output type of the one matches the input type of the next.

We now turn to the monadic translation itself.
In both examples, we handled multiple possible sets of effects by using a monad $\Monad{\e}(-)$ indexed on $\e \in \Eff$.
We then required a translation that took a well-typed effectful program $e$ where $\Gamma \proves e : \tau \ctxsep \e$
to a well-typed pure program $e'$ where $\Gamma \proves e' : \Monad{\e}(\tau)$.
The \textsc{Capture} rule incorporates this translation, which we denote $\Captured{-}{-}$, into our framework.
\footnote{\citet{Tate13}~calls capture \emph{thunking}, to bring attention to the similarity with the familiar concept in functional languages.
  We use the word ``capture'' here because it better-fits how we use the concept.}
Note that the pure and effectful languages may use different type constructors.
For instance, our example effectful languages annotated function types with effects, but DCC does not.
While we implicitly translated $\tau_1 \effto \tau_2$ to $\tau_1 \to \Monad{\e}(\tau_2)$ in Sections~\ref{sec:state-and-exns} and~\ref{sec:tsni},
our general framework makes this translation explicit and denotes it $\pureTypeTrans{-}$.

The lattice structure of $\Eff$ coupled with the $\eta$, $\BindM{}$, and $\CoerceName$ operations meant that $\Monad{\e}(-)$ formed an indexed monad~\citep{WadlerT98,OrchardPM14}.
As our framework generalizes $\Eff$ to an effector, it correspondingly generalizes $\Monad{\e}$ to a \emph{productor}~\citep{Tate13},
a generalization of both indexed monads and graded monads~\citep{Katsumata14,FujiiKM16}, which require $\Eff$ to be an ordered monoid.

Unsurprisingly, productors require structure similar to $\eta$, $\BindM{}$, and $\CoerceName$, which we specify in the \textsc{Map} and \textsc{Join} rules.
\textsc{Map} requires, for each effect~$\e$, a pure program transformer $\MapE{\e}$
that takes a program from $\pureColor{\tau}$ to $\pureColor{\tau'}$ and produces a program from $\pureColor{\CapturedType{\e}{\tau}}$ to $\pureColor{\CapturedType{\e}{\tau'}}$.
\textsc{Join} requires a program $\JoinE{\e_1, \dotsc, \e_n}{\e}$ whenever $\ComposeE{\e_1, \dotsc, \e_n}{\e}$
that translates a pure program capturing effects $\e_1$ through $\e_n$ in order into a pure program that captures effect $\e$.
Notably, every indexed monad is also a productor and any productor whose effector $\Eff$ forms a lattice is also an indexed monad~\citep{Tate13}.

\textsc{CapturedSeq} says that we can define the capture of the sequential composition of effectful programs by capturing each program individually and using pure composition, $\MapE{}$, and $\JoinName$ appropriately.
Here $\pureColor{\rho_1 \progEq{\tau_1, \tau_2} \rho_2}$ means ``$\rho_1$ and $\rho_2$ are equal as pure programs from $\tau_1$ to $\tau_2$.''
Those familiar with monads will recognize this as a generalization of monadic bind.
Indeed, a monad (or indexed monad) requires $\ComposeE{\e,\e}{\e}$ for all $\e \in \Eff$, allowing us to define $\BindM{\e}(\rho) = \MapE{\e}(\rho) \seq \JoinE{\e,\e}{\e}$.
Those familiar with category theory will also recognize this as Klesili composition.

Our framework up to this point is developed directly from \citet{Tate13} and gives rules defining the semantics of a type-and-effect system using a productor.
However, giving semantics to program-counter labels also requires our effects to be \emph{secure}.
For that, we need two more laws.

In Sections~\ref{sec:state-and-exns} and~\ref{sec:tsni} we went to great pains to ensure that the labels on our pure monads properly represented the visibility of our effects.
In Section~\ref{sec:tsni} we did this directly with the protection relation on a pointed type $\LiftType{\tau}$.
In Section~\ref{sec:state-and-exns} it was a bit more complicated.
We required that $\StateLabel \prot \sigma$, the type of the state cell, and $\Monad{\ExnEff}(-)$ explicitly wrapped its output with $\ExnLabel$ to constrain visibility.
The end goal was that, if both the output and effects of the original program had sensitivity at least~$\ell$, then the output of the translated pure program must also have sensitivity at least~$\ell$.
In order to codify that goal in our framework, we assume both the pure and effectful languages have a protection relation defining sensitivity and relate them in exactly this way.
We then require \textsc{ProtectC} to hold in our framework, which is exactly the rule that we used in Section~\ref{sec:tsni}.
We can furthermore show that \textsc{ProtectC} holds in the example from Section~\ref{sec:state-and-exns} by case analysis on the effects $\varepsilon$, using the fact that $\ExnLabel \flowsto \StateLabel$ and that $\StateLabel \prot \sigma$.

Note that, while we assume that both languages define a protection relation, it does not have to be the main security mechanism of the language.
Instead, protection can be defined in terms of that security mechanism.
For instance, a fine-grained system can define a protection relation using the fact that every type has an associated label.
\citet{RajaniG18}~give an example of this sort of definition in their translation of a fine-grained system into a coarse-grained one.

In Section~\ref{sec:state-and-exns} we used a well-known monadic translation that creates a simulation.
That is, if $\langle p, s \rangle \stepsone \langle p', s' \rangle$, then $(\Captured{\e}{p}~s) \stepsmany (\Captured{\e}{p'}~s')$.
This property tells us that any context that distinguishes $p$~from~$q$
translates to a context that distinguishes $\Captured{\varepsilon}{p}$~from~$\Captured{\varepsilon}{q}$.
In other words, effectful programs must be indistinguishable whenever their captured pure counterparts are.
Because noninterference specifies that an attacker \emph{cannot} distinguish programs whose outputs (including effects) are highly sensitive,
this property allows us to lift noninterference from the pure language to the effectful one.
While encoding operational semantics, contextual equivalence, and simulation into our framework would require a lot of work, we can directly encode this last insight.
We do so with \textsc{EquivCap}.
We require two equivalence relations parameterized on a label~$\ell$:~one for pure programs and one for effectful ones.
Intuitively, two programs are equivalent at~$\ell$ if they are indistinguishable to attacker who can read values only up to level~$\ell$.
The equivalence of effectful programs also takes an effect parameter, allowing the definition to account for visible effects.
\textsc{EquivCap} demands only that they capture the correspondence we relied on above.

\textsc{EquivCap} is also the rule that requires the monadic translation to capture the semantics of the effect.
For instance, if we replace the monadic translation in Section~\ref{sec:state-and-exns} with a faulty one, we will not enjoy \textsc{EquivCap} with contextual equivalences.
Imagine in particular a translation that treated the state as a constant and discarded writes.
The effectful programs $\Write(3)$ and $\Write(4)$ would then be distinguishable by a context that reads the state and compares it to $3$,
but they would translate to identical---and therefore indistinguishable---pure programs.
Of course, we could change the notion of effectful equivalence as well, essentially changing the semantics of our effectful language.
However, as we will see in Section~\ref{sec:type-and-effect-ni}, that equivalence defines the meaning of noninterference in our framework,
so the guarantees and semantics given by our framework will reflect this new equational semantics.

With these rules in place, the framework allows us to prove a strong and general security theorem for effectful programming languages.
Moreover, it connects to the category theory, giving us powerful tools for reasoning about effects.

\subsection{Better State and Exceptions via Partialiaty}
\label{sec:bett-state-except}

In Section~\ref{sec:state-and-exns} we avoided a concern about state and exceptions combining to leak data by assuming that $\ExnLabel \flowsto \StateLabel$.
Recall that when we compose a program~$p$ that may throw an exception with a program~$q$ that may write state, the write---or lack thereof---can leak whether $q$ executed and thus whether $p$ threw an exception.
We therefore restricted our system to require that anyone who could observe the state could also see any exceptions.

We now modify our rules from Section~\ref{sec:tracking-effects-state-exns} to remove this restriction while still ensuring that state and exceptions interact securely.
We use the translation from Section~\ref{sec:effectf-nonint-free} as a guide.
When composing programs~$p$ and~$q$ where $p$ may throw an exception and $q$ has effect $\e$, $\Captured{}{p}$~wraps its output in $\ExnLabel$.
This represents the fact that, to read that data, an attacker must be able to see whether an exception has been thrown (since if one had, that data would not exist).
To compose these programs, then, anyone who can see the results of any effects in~$q$ must be able to see whether an exception occurred.
In other words, for the translation to be well-typed, we require $\ExnLabel \flowsto \ell_\e$.

By assuming $\ExnLabel \flowsto \StateLabel$, we assumed an adversary who could see the results of any effect could see if an exception was thrown.
This made the requirement above trivial,
but it prevented the language from representing e.g., writes that were more secret than exceptions.
We could alternatively enforce noninterference by restricting composition of effects as follows:
\begin{mathpar}
  \infer*{
    \e_1 \cup \e_2 \subseteq \e_3 \\
    E \notin \e_1
  }{\ComposeE{\e_1, \e_2}{\e_3}}
  \and
  \infer*{
    \e_1 \cup \e_2 \subseteq \e_3 \\
    \ExnLabel \flowsto \ell_{\e_2}
  }{\ComposeE{\e_1, \e_2}{\e_3}}
\end{mathpar}

Not that this makes the composition relation partial, since not all pairs of effects compose.
This means that $\Eff$ is no longer a lattice or even an ordered monoid.
As a result, the translation we defined in Section~\ref{sec:effectf-nonint-free} is no longer an indexed monad, or even a graded monad.
This solution relies on the generality of effectors and productors.

This precise modification to the composition rules---and the partiality that results---is critically important for more realistic languages.
For example, FlowCaml~\citep{PottierS02} includes multiple exception types and general mutable reference cells with different types.
In that setting, a program~$p$ that may throw an exception with label~$\ell$ can safely compose with a program~$q$ that writes data exclusively at or above label~$\ell$.

\section{The Noninterference Half-Off Theorem}
\label{sec:niho}

Our main aim is to show how, given a \pc system, we can give semantics to the \pc label and prove noninterference for that system.
We do this by translation to a type-and-effect system, which we give semantics using the framework from Section~\ref{sec:prod-effects-ifc}.
This is a powerful and general result, but it requires us to first demonstrate the security of the type-and-effect system itself.

\subsection{Type-and-Effect Noninterference}
\label{sec:type-and-effect-ni}

In Sections~\ref{sec:state-and-exns}~and~\ref{sec:tsni} we leveraged the ability to translate effects into a pure language to simplify reasoning about noninterference.
This allowed us to prove noninterference for the effectful language while only proving it directly for the pure part of the language.
Since the framework of Section~\ref{sec:prod-effects-ifc} specifies when this is possible, we would like to prove noninterference for any languages which admit the rules in Figure~\ref{fig:framework-rules}.

As in our examples, noninterference formalizes the intuition that adversary at label $\AtkLabel$ can only distinguish data and effects at or below label $\AtkLabel$.
We again define the label of data using a protection relation, though this time we leave the details of that relation abstract.
We also assign a label $\ell_\e$ to the effect $\e$ to represent $\e$'s sensitivity.
Finally, as each language has a different notion of equivalence, use an abstract notion of equivalence $\equiv_{-}$ parameterized on labels.
\begin{defn}[Abstract Noninterference]
  \label{defn:abstract-ni}
  Let~$r$ be a program such that $\effColor{t_1 \provesEff r \dashv t_2 \ctxsep \e_1}$.
  We say that~$r$ is \emph{noninterfering with respect to $\equiv_{-}$} if, for all labels $\ell \in \Labs$ and programs $p$ and $q$ such that
  \begin{enumerate}
    \item \label{absni:li:inputs} $\effColor{t_3 \provesEff p \dashv t_1 \ctxsep \e_2}$ and $\effColor{t_3 \provesEff q \dashv t_1 \ctxsep \e_2}$ with $\ComposeE{\e_2, \e_1}{\e}$
    \item \label{absni:li:protection} $\effColor{\ell \prot t_1}$ and $\ell \flowsto \ell_{\e_2}$
  \end{enumerate}
  then for all labels $\AtkLabel \in \Labs$, either $\ell \flowsto \AtkLabel$ or $\effColor{p \seq r \equiv_{\AtkLabel} q \seq r}$.
\end{defn}
In the above definition, condition~\ref{absni:li:inputs} requires the sequential compositions $p \seq r$ and $q \seq r$ to be well-typed, while $p$ and $q$ each produce effects (at most) $\e_2$.
This sequential composition represents providing two different inputs to the program~$r$, abstracting the contextual equivalence we used in Sections~\ref{sec:pure-ifc-lang}, \ref{sec:state-and-exns}, and~\ref{sec:tsni}.
Condition~\ref{absni:li:protection} requires that the sensitivity of both the type, $t_1$, and the effects, $\e_2$, of the input programs be at least $\ell$.
Intuitively, the conclusion says that an attacker can only use $r$ to distinguish between $p$ and $q$ if they could already see the effects or outputs of $p$ and $q$, and thus distinguish them without $r$.

We also allow the definition to apply to pure programs, replacing type-and-effect judgements with pure judgements and disregarding other references to effects.
Our framework's rules are then sufficient to transfer a noninterference result from pure programs to effectful ones.
\begin{thm}[Type-and-Effect Noninterference]
  \label{thm:type-and-effect-ni}
  For any system satisfying all rules in Figure~\ref{fig:framework-rules} where every well-typed pure program is noninterfering with respect to $\lequiv[-]$,
  then every program well-typed in the type-and-effect system is noninterfering with respect to $\lequivEff[-]{\e}$.
\end{thm}

\begin{proof}
  Unfolding Definition~\ref{defn:abstract-ni}, we have programs $p$, $q$, and $r$ and a label $\ell$ such that
  \begin{enumerate}
    \item $\effColor{t_1 \provesEff r \dashv t_2 \ctxsep \e_1}$,
    \item $\effColor{t_3 \provesEff p \dashv t_1 \ctxsep \e_2}$ and $\effColor{t_3 \provesEff q \dashv t_1 \ctxsep \e_2}$ with $\ComposeE{\e_2, \e_1}{\e}$,
    \item $\effColor{\ell \prot t_1}$ and $\ell \flowsto \ell_{\e_2}$,
  \end{enumerate}
  and we aim to show that for all labels~$\AtkLabel \in \Labs$, either $\ell \flowsto \AtkLabel$ or $\effColor{p \seq r \lequivEff[\AtkLabel]{\e} q \seq r}$.

  Let $\rho = \MapE{\e_2}(\Captured{\e_1}{r})\mkern-2mu \seq \JoinE{\e_2, \e_1}{\e}$.
  The rules in Figure~\ref{fig:framework-rules} guarantee $\pureColor{\CapTypePref{\e_2} \pureTypeTrans{t_1} \provesPure \rho \dashv \CapTypePref{\e} \pureTypeTrans{t_2}}$
  and \textsc{CapturedSeq} requires \mbox{$\pureColor{\Captured{\e}{p \seq r} \progEq{\pureTypeTrans{t_3}, \CapTypePref{\e} \pureTypeTrans{t_2}} \Captured{\e_2}{p}\mkern-2mu \seq \rho}$}, and similarly for $q$.
  All well-typed pure programs are noninterfering with respect to $\lequiv[-]$, $\rho$ is pure, and $\pureColor{\ell \prot \CapTypePref{\e_2}\pureTypeTrans{t_1}}$ by \textsc{ProtectC}.
  Thus, for any label~$\AtkLabel$, either $\ell \flowsto \AtkLabel$ or $\pureColor{\Captured{\e_2}{p} \seq \rho \lequiv[\AtkLabel] \Captured{\e_2}{q} \seq \rho}$.
  In the first case, we are finished.
  In the second case, the program equality above and \textsc{EquivCap} give us
  \[
    \pureColor{\Captured{\e_2}{p}\mkern-2mu \seq \rho \lequiv[\AtkLabel] \Captured{\e_2}{q}\mkern-2mu \seq \rho}
    \iff
    \pureColor{\Captured{\e}{p \seq r} \lequiv[\AtkLabel] \Captured{\e}{q \seq r}}
    ~\Longrightarrow~
    \effColor{p \seq r \lequivEff[\AtkLabel]{\e} q \seq r}.
    \qedhere
  \]
\end{proof}

Theorem~\ref{thm:type-and-effect-ni} lifts noninterference of pure programs to effectful programs when the corresponding notions equivalence satisfy \textsc{EquivCap}.
We can now see what happens if these equivalences do not match expectations.
Recall our example from Section~\ref{sec:prod-effects-ifc}: we translate state by discarding writes and returning a constant for all reads.
\textsc{EquivCap} no longer holds for contextual equivalence, but it does hold for other equivalences.
There may be many such equivalences, but one simple option is the trivial effectful equivalence that is always true.
Using this equivalence, our example now admits all rules in Figure~\ref{fig:framework-rules}, so Theorem~\ref{thm:type-and-effect-ni} applies.
However, we are now giving trivial semantics to the type-and-effect system.
Abstract noninterference with respect to this semantics merely says that an attacker who cannot distinguish anything cannot distinguish sensitive programs.
This result is both intuitively and technically trivial.
The instantiation of the framework, while allowed, is therefore probably uninteresting.

\subsection{Semantics of Program Counter Labels}

\begin{figure}
  \begin{mathpar}
    \infer*[left=PcEff]{\pcColor{t \ctxsep \pc \provesPc p \dashv t'}}{\exists \varepsilon. \effColor{\effTypeTrans{t} \provesEff p \dashv \effTypeTrans{t'} \ctxsep \varepsilon}}
    \and
    \infer*[left=ProtectTrans]{\pcColor{\ell \prot t}}{\effColor{\ell \prot \effTypeTrans{t}}}
  \end{mathpar}
  \caption{Additional Rules for \pc systems}
  \label{fig:pc-system-rules}
\end{figure}

We can now use the semantic framework we have developed for effectful labeled programs and noninterference for type-and-effect systems to talk about the semantics and security of the $\pc$ label.
We extend the framework to include a \pc system with judgments of the form $\pcColor{t \ctxsep \pc \provesPc p \dashv t'}$.
(We still use Roman letters for types and programs in the \pc~system, but we color them in \pcColor{green}.)
Figure~\ref{fig:pc-system-rules} shows the rules we require for this extended framework.

We give semantics to the \pc by formalizing the intuition that it constrains programs to only use secure effects.
Specifically, we define the semantics by requiring a translation of typing proofs in the \pc system to typing proofs in the type-and-effect system, which guarantees security by Theorem~\ref{thm:type-and-effect-ni}.
\textsc{PcEff} formalizes this requirement.

For this semantics to make sense, we would like it to preserve types.
Unfortunately, in our examples, the \pc systems and type-and-effect systems had different function types.
The \pc system included a label on its functions ($\tau_1 \pcto \tau_2$), while the type-and-effect system included an effect ($\tau_1 \effto \tau_2$).
We therefore allow the \pc system to have different types, but the same programs, and require there to be a translation $\effTypeTrans{-}$ from the \pc types to the type-and-effect types.
This translation must preserve the sensitivity of the data, represented as the protection level, a requirement we formalize as rule~\textsc{ProtectTrans}.

These rules complete the requirements for our core theorem.
\begin{thm}[The Noninterference Half-Off Theorem]
  \label{thm:niho}
  For any system satisfying all rules in Figures~\ref{fig:framework-rules} and~\ref{fig:pc-system-rules}, if every well-typed pure program is noninterfering with respect to $\lequiv[-]$,
  then every effectful program well-typed in the \pc system is noninterfering with respect to $\lequivEff[-]{\e}$.
\end{thm}

\begin{proof}
  This follows directly from \textsc{PcEff} and Theorem~\ref{thm:type-and-effect-ni}.
  We note that this always instantiates the~$\e$ in $\lequivEff[-]{\e}$ with the same~$\e$ used to type check~$r$.
\end{proof}

\section{Deepening the PC-Effect Connection}
\label{sec:deepening-pc-effect}

So far we have kept the connections between effects and program-counter labels lightweight: we only required a function $\ell_{-}$ from effects to labels and the two rules from Figure~\ref{fig:pc-system-rules}.
This means that our framework can give semantics to many systems.
This generality, however, prevents us from proving some interesting theorems which we would like to prove.
In this section, we strengthen the connection between program-counter labels and effects, allowing us to prove stronger results.

In particular, we formalize the aphorism that the \pc is a lower bound on effects.
Interestingly, not every language that fits our framework treats the \pc as a lower bound on effects, despite the fact that they are secure by Thoerem~\ref{thm:niho}.
Indeed, a rather simple counterexample shows that the \pc can always be incomparable to the label of an effect.

Still, all of our realistic examples do treat the \pc as a lower bound on effects.
We show that this is because they admit a few simple rules on top of the framework we have developed so far.
Moreover, in all of our examples so far, we can extend the function~$\ell_{-}$ into a \emph{Galois connection} between labels and effects.
Beyond being intrinsically interesting, it also allows us to refine our formalization of the folklore above, producing a more-concrete result.

\subsection{Is the PC a Lower Bound on Effects?}
\label{sec:pc-lower-bound}

We start by formalizing the folklore statement that ``the \textsf{pc} is a lower bound on effects.''
As mentioned in Section~\ref{sec:introduction}, taken literally this aphorism does not even seem to type-check, since we are trying to bound an effect by a label.
However, we can use our function~$\ell_{-}$ to formalize the statement by saying that the $\pc$ bounds the \emph{label} of the effect.
\begin{defn}[$\pc$-Bounded Effects]
  \label{defn:pc-bound-eff}
  We say effects are \emph{$\pc$-bounded} if whenever $\pcColor{t \ctxsep \pc \provesPc p \dashv t'}$,
  there is some $\varepsilon$ such that $\effColor{\effTypeTrans{t} \provesEff p \dashv \effTypeTrans{t'} \ctxsep \varepsilon}$ where $\pc \flowsto \ell_\varepsilon$.
\end{defn}
Note that Definition~\ref{defn:pc-bound-eff} only requires that the \pc bound \emph{some}~$\varepsilon$.
We might at first think that any~$\varepsilon$ with which $p$ type-checks in the type-and-effect system should be bounded below by \pc,
but effect variance prevents that definition from applying to most languages.
To see why, imagine a program~$p$ in the system from Section~\ref{sec:state-and-exns} that reads some data from state, but can never write data or throw an exception.
It can type-check with a $\pc$ of $\top$, which flows to $\ell_\REff$, since $\ell_\REff = \top$.
However, by variance $p$ can also type-check with effect $\{\REff,\WEff,\ExnEff\}$, and $\top \not\flowsto \ell_{\{\REff, \WEff, \ExnEff\}}$.
The existential quantifier in Definition~\ref{defn:pc-bound-eff} thus provides a meaningful statement while allowing imprecision due to variance.

Common wisdom suggests that any language that uses a \pc to enforce noninterference should have \pc-bounded effects.
However, this is not the case, as we can show using our framework.

Consider a language with state and exceptions, based on that from Section~\ref{sec:state-and-exns}.
In the original language, any preorder could serve as the set of information-flow labels (though we used a join semilattice for convenience).
However, in the new language we will use a join semilattice of a special form.
Intuitively, we want two equivalent but unrelated spaces of labels, one for effects and one for program-counter labels.
Hence, we use a \emph{semilattice coproduct}: given a semilattice of labels~$\Labs$ sufficient to represent our effects, we construct a new semilattice with two disjoint copies of~$\Labs$.
We cannot work directly over the disjoint union~$\Labs + \Labs$, since this is not a semilattice---there is no join of two labels $\Inl(\ell_1)$ and $\Inr(\ell_2)$.
However, if we add a new distinguished top element, the result is a semilattice.
In fact, it is the smallest semilattice that contains two disjoint copies of~$\Labs$.
Thus, we use a semilattice of this form for this example.

The modified language differs from the original in three ways:
(\one\/) the typing rule for \Labeled, (\two\/) the typing rules for effectful operations, and (\three\/) the function $\ell_{-}$.
First, the \Labeled rule now forces all labels into the left-hand side of the lattice.
That is, the rule is split into three cases:
\begin{mathpar}
\infer{\Gamma \proves e : \tau}{\Gamma \proves \LabeledProgram{\Inl(\ell)}{e} : \LabeledType{\Inl(\ell)}{\tau}} \and
\infer{\Gamma \proves e : \tau}{\Gamma \proves \LabeledProgram{\Inr(\ell)}{e} : \LabeledType{\Inl(\ell)}{\tau}} \and
\infer{\Gamma \proves e : \tau}{\Gamma \proves \LabeledProgram{\top}{e} : \LabeledType{\Inl(\top)}{\tau}} \and
\end{mathpar}

Second, we also consider functions~$\ell_{-}$ of a special form.
Intuitively, the lattice of labels has a label space on the left for data, and a label space on the right for effects.
We thus need $\ell_{\varepsilon}$ to always be of the form $\Inr(\ell)$ for some $\ell \in \Labs$.
To do this, we pick a function $\hat{\ell}_{-} : \Eff \to \Labs$ connecting effects to the original semilattice $\Labs$, such as the effect-to-label function we used in Section~\ref{sec:state-and-exns}.
We then lift~$\hat{\ell}_{-}$ to the full label space by defining $\ell_\e = \Inr(\hat{\ell}_\e)$.

Finally, we modify the rules that compare the $\pc$ and effect labels by explicitly comparing the $\pc$ to the data-label analogue of the effect's label.
Formally, we use the following rules:
\begin{mathpar}
  \infer{
    \Gamma \ctxsep \pc \proves e : \sigma \\
    \pc \flowsto \Inl(\hat{\ell}_\WEff)
  }{\Gamma \ctxsep \pc \proves \Write(e) : \Unit}
  \and
  \infer{
    \pc \flowsto \Inl(\hat{\ell}_\ExnEff)
  }{\Gamma \ctxsep \pc \proves \Throw : \tau}
\end{mathpar}
All of the other rules remain unchanged from those in Section~\ref{sec:state-and-exns}.

This fits our framework and in fact admits exactly the same programs as the original $\pc$-based system did.
However, if a program type-checks with some $\pc$, $\pc = \Inl(\ell)$ for some~$\ell$, while the effect label will be $\ell_{\varepsilon} = \Inr(\ell')$ for some~$\ell'$.
By construction, we cannot have $\Inl(\ell) \flowsto \Inr(\ell')$ for any labels~$\ell$ and~$\ell'$.

This example shows that it is possible to have a secure language in our framework where the \pc and the label of the effect are incomparable.
The language is noninterfering, yet its effects are not \pc-bounded.
However, we only need a few simple additions to our framework to ensure that a language's effects are \pc-bounded.

Consider a program~$p$ in one of our example type-and-effect systems such that $\Gamma, x \ty \tau \proves p : \tau' \ctxsep \varepsilon$.
We can transform this into a program on labeled data by unlabeling the input, running $p$, and labeling its output.
That is, we can build a program transformer~$\MapE{\ell}$ that we can type-check as $\Gamma, x \ty \LabeledType{\ell}{\tau} \proves \MapE{\ell}(p) : \LabeledType{\ell}{\tau'} \ctxsep \varepsilon$.
Notably, this does not change the effect.
To retain security, we must ensure that $\ell \flowsto \ell_\varepsilon$, since $p$ may otherwise leak data about the $\ell$-labeled input in its effects.

When a program of the form $\MapE{\ell}(p)$ has effect~$\varepsilon$---that is, $\Gamma, x \ty \LabeledType{\ell}{\tau} \proves \MapE{\ell}(p) : \LabeledType{\ell}{\tau'} \ctxsep \varepsilon$---we know that the effect $\varepsilon$ must come from~$p$.
Moreover, we know that $p$ must have type-checked with some effect~$\varepsilon'$ such that $\ell \flowsto \ell_{\varepsilon'}$.
However, this need not be $\varepsilon$, due to similar reasoning about variance that we saw in the design of Definition~\ref{defn:pc-bound-eff}.
Again, this leads to an existential quantifier.

The program transformer~$\MapE{\ell}(p)$ also has a similar action in the \pc system as it did in the type-and-effect system.
If $\Gamma, x \ty \tau \ctxsep \pc \proves p : \tau'$, then we want $\Gamma, x \ty \LabeledType{\ell}{\tau} \ctxsep \pc \proves \MapE{\ell}(p) : \LabeledType{\ell}{\tau'}$.
However, now the limiter is the \pc rather than the effect.
That is, this only type-checks if $\ell \flowsto \pc$.
We also note that $\MapE{\ell}$ type-checks in both the \pc and type-and-effect systems because the type translation between them leaves labels alone:
$\effTypeTrans{\LabeledType{\ell}{\tau}} = \LabeledType{\ell}{\effTypeTrans{\tau}}$.

\begin{figure}
  \centering
  \begin{mathpar}
    \infer*[left=MapPC]{\pcColor{t \ctxsep \pc \provesPc p \dashv t'}}{\pcColor{\LabeledType{\pc}{t} \ctxsep \pc \provesPc \MapE{\pc}(p) : \LabeledType{\pc}{t'}}}\and
    \infer*[left=MapEffInv]{\effColor{\LabeledType{\ell}{t} \provesEff \MapE{\ell}(p) \dashv \LabeledType{\ell}{t'} \ctxsep \varepsilon}}{\exists \varepsilon'.\,\effColor{t \provesEff p \dashv t' \ctxsep \varepsilon'} \land \ell \flowsto \ell_{\varepsilon'}} \and
    \infer*[left=LabTrans]{}{\effColor{\effTypeTrans{\pcColor{\LabeledType{\ell}{t}}}} = \effColor{\LabeledType{\ell}{\effTypeTrans{\pcColor{t}}}}}
  \end{mathpar}

\caption{Additional Rules for \pc as a Lower Bound}
\label{fig:pc-rules}
\end{figure}

We provide version of these rules in Figure~\ref{fig:pc-rules} using the single-input, single-output judgments from Sections~\ref{sec:prod-effects-ifc} and~\ref{sec:niho}.
These versions are, in fact, slightly more general.
First, the \textsc{MapPC} rule only requires that we be able to map the current \pc, rather than any label~$\ell$ where~$\ell \flowsto \pc$.
Second, we have no requirement that all effects of $\MapE{\ell}(p)$ come from $p$, only that the label of the effects that do come from $p$ are bounded below by $\ell$.

Note also that we do not have a rule corresponding to $\MapE{\ell}$ in the type-and-effect system.
\textsc{PcEff} will ensure that $\MapE{\ell}$ has the right type, which is all we need for the applications in this paper.
However the rule \textsc{MapEffInv} is very suggestive, and most any language that admits \textsc{MapEffInv} will also admit an appropriate rule for $\MapE{\ell}$ for effects.

Adding these assumptions to our framework is sufficient to prove that a language is \pc-bounded.

\begin{thm}
  \label{thm:pc-to-te}
  The effects of any language admitting the rules in Figures~\ref{fig:framework-rules}, \ref{fig:pc-system-rules}, and~\ref{fig:pc-rules} are \pc-bounded.
\end{thm}
\begin{proof}
\[
  \newsavebox{\LabelBox}
  \savebox{\LabelBox}{\small\textsc{PcEff} and \textsc{LabTrans}}
  \infer*[Right=MapEffInv]{
    \infer*[Right=\usebox{\LabelBox}]{
      \infer*[Right=MapPC]{
        \pcColor{t \ctxsep \pc \provesPc p \dashv t'}
      }{\pcColor{\LabeledType{\pc}{t} \ctxsep \pc \provesPc \MapE{\pc}(p) \dashv \LabeledType{\pc}{t'}}}
    }{\exists \varepsilon.\,\effColor{\LabeledType{\pc}{\effTypeTrans{t}} \provesEff \MapE{\pc}(p) \dashv \LabeledType{\pc}{\effTypeTrans{t'}} \ctxsep \varepsilon}}
  }{\exists\varepsilon'.\,\effColor{\effTypeTrans{t} \provesEff p \dashv \effTypeTrans{t'} \ctxsep \varepsilon'} \land \pc \flowsto \ell_{\varepsilon'}}
  \qedhere
\]
\end{proof}

\subsection{Computing PC Bounds via Galois Connections}
\label{sec:galo-conn-their}

Our example languages have even more structure: we know what effect a program has based on its \pc.
That is, we can build a function $\gamma$ from labels to effects such that if a program type-checks with program-counter label \pc then it type-checks with effect $\gamma(\pc)$.
Let us examine this in detail for the example language from Section~\ref{sec:state-and-exns}.
We used a function~$\ell_{-}$ on effects such that $\ell_{\{E\}} = \ExnLabel$, $\ell_{\{R\}} = \top$, and $\ell_{\{W\}} = \StateLabel$ and on arbitrary sets it acts as a lower bound.
Note that because of this $\ell_{-}$ must be \emph{antitone}: if $\varepsilon \subseteq \varepsilon'$, then $\ell_{\varepsilon'} \flowsto \ell_{\varepsilon}$.

Given a particular $\ell_{-}$, we can then define the function $\gamma$ from labels to effects as follows:$$\gamma(\ell) = \{R\} \cup \{W \mid \ell \flowsto \StateLabel\} \cup \{E \mid \ell \flowsto \ExnLabel\}$$

The functions~$\ell_{-}$ and~$\gamma$ form a \emph{Galois connection}.
Galois connections are well-known for their uses in abstract interpretation~\citep{CousotC77}.
However, we will see that they can be used here to strengthen Theorem~\ref{thm:pc-to-te} by providing a witness for the existential in Definition~\ref{defn:pc-bound-eff}.

Since $\ell$ is antitone, it seems like it cannot be part of a Galois connection, as Galois connections are defined on \emph{monotone} functions.
However, it turns out $\ell_{-}$ and $\gamma$ more precisely form an \emph{antitone} Galois connection~\citep[see~e.g.,][]{Galatos07}.
An antitone Galois connection between lattices~$A$ and~$B$ is equivalent to a monotone Galois connection between $A$ and $B^{\text{op}}$, the order dual of $B$.

\begin{lem}[Antitone Galois Connection]
  The functions~$\ell_{-}$ and~$\gamma$ form an antitone Galois connection.
  That is, for any label $\ell$ and effect $\varepsilon$, $\ell \flowsto \ell_{\varepsilon}$ if and only if $\varepsilon \subseteq \gamma(\ell)$.
\end{lem}
\begin{proof}
  By examining all of the (eight) possible values of $\varepsilon$.
\end{proof}
\noindent A similar construction and proof can be done for the example of Section~\ref{sec:tsni}, and for realistic languages like Jif~\citep{jif-release35}, Fabric~\citep{LiuAGM17}, and FlowCaml~\citep{PottierS02}.

With this structure, it becomes relatively easy to strengthen Theorem~\ref{thm:pc-to-te}.
Here we use $\ComposeE{\e}{\e'}$ as the general ordering relation on effects.
\begin{thm}
  If a language admits the rules in Figures~\ref{fig:framework-rules}, \ref{fig:pc-system-rules}, and \ref{fig:pc-rules},
  and there is a function $\gamma$ such that $\ell_{-}$~and~$\gamma$ form an antitone Galois connection,
  then whenever $\pcColor{t \ctxsep \pc \provesPc p \dashv t'}$, then $\effColor{\effTypeTrans{t}  \provesEff p \dashv \effTypeTrans{t'} \ctxsep \gamma(\pc)}$.
\end{thm}
\begin{proof}
  By Theorem~\ref{thm:pc-to-te}, we know that there is some $\varepsilon$ such that $\effColor{\effTypeTrans{t} \provesEff p \dashv \effTypeTrans{t'} \ctxsep \varepsilon}$ and $\pc \flowsto \ell_\varepsilon$.
  \mbox{Because $\pc \flowsto \ell_\varepsilon$} and $\ell_{-}$ and $\gamma$ form an antitone Galois connection, $\ComposeE{\varepsilon}{\gamma(\pc)}$.
  This means we can further apply $\textsc{Seq}_\leq$ to get $\effColor{\effTypeTrans{t} \provesEff p \dashv \effTypeTrans{t'} \ctxsep \gamma(\pc)}$.
\end{proof}

\section{Related Work}
\label{sec:related}

This work pulls mostly from two distinct areas: static IFC and the theory of effects.
We discuss related work from each of these areas in turn.

\subsection{Static Information-Flow Control}
\label{sec:related-ifc}

Noninterference, originally introduced by \citet{GoguenM82}, is the foundational security property of information-flow control systems.
While originally proposed to avoid confidentiality leaks, noninterference can apply to any security policy expressible as a preorder of labels.
Since \citeposessive{VolpanoSI96} seminal work enforcing noninterference with a type system,
numerous others have used type systems to guarantee noninterference for functional and imperative languages, with and without effects,
where security policies represent confidentiality, integrity, or even distributed consistency~\citep{HeintzeR98,AbadiBJHR99,PottierS02,ZdancewicM02,SabelfeldM03,TsaiRH07,RafnssonS14,MilanoM18,VassenaRBW18}.

Termination is one channel many type-based enforcement mechanisms ignore~\citep[e.g.,][]{VolpanoSI96,PottierS02,jif-release35,LiuAGM17}.
As \citet{VolpanoS97} showed and we observed in Section~\ref{sec:tsni}, enforcing termination-sensitive noninterference with a type system is possible, but highly restrictive.
Unfortunately, \citet{AskarovHSS08} argue that termination channels can leak an arbitrary amount of data, making it dangerous to ignore them.
We hope that our framework's ability to unify possible nontermination with other effects will connect to
recent work on precisely specifying and constraining leakage through termination in more permissive languages~\citep{MooreAC12,BayA20}.

Prior work uses a wide variety of techniques to prove noninterference.
The first proofs of static noninterference~\citep{VolpanoSI96,VolpanoS97,HeintzeR98} relied on structural induction with careful manual reasoning.
\citet{PottierS02} used bracketed pairs of terms to simulate two program executions with different high inputs and compare the outputs.
This technique makes combining state and exceptions tractable, but provides no means to reason about termination.
Other proofs rely on semantics using partial equivalence relations~\citep{AbadiBJHR99,TseZ04,SabelfeldS01} or logical relations~\citep{ShikumaI08,RajaniG18}.
The complexity of all of these approaches lies in reasoning about effects, demonstrating the value of noninterference half-off.

We base all of the example languages in this paper on DCC~\citep{AbadiBJHR99}.
DCC was originally designed to explore dependency, with information flow as an interesting special case.
Interestingly, DCC was not given an operational semantics or a noninterference theorem in the original paper.
Instead, \citet{AbadiBJHR99} described a domain-theoretic semantics, and used it to prove a semantic security theorem closely related to noninterference.
\citet{TseZ04} later developed an operational semantics for DCC, and claimed to prove a noninterference theorem analogous to the one we used in Section~\ref{sec:pure-ifc-lang} by translating DCC into System~F and using parametric reasoning.
\citet{ShikumaI08}, however, found a flaw in \citeauthor{TseZ04}'s proof, which \citet{BowmanA15} later repaired.
\citet{AlgehedB19} extended and simplified the proof technique, leading to a verified version of the proof written in Agda.

DCC is the paradigmatic \emph{coarse-grained} IFC language, a style that is characterized by labeling and unlabeling data.
The various other results mentioned above all employ \emph{fine-grained} IFC systems, where each type includes a label.
Though the two approaches may appear substantially different, \citet{RajaniG18} proved them equivalent.

Both DCC\textsubscript{pc}~\citep{TseZ04} and the Sealing Calculus~\citep{ShikumaI08} include \emph{protection context labels} that look very similar to our \pc~labels.
Both languages, however, are pure, and the labels serve only to securely include a more permissive typing rule for \Unlabel.
While our examples could also employ this technique, it would increase the complexity of the type systems,
particularly the type-and-effect systems which would need to include both protection context labels and effects.

Other work has implemented coarse-grained IFC as monadic libraries, mostly in Haskell~\citep{VassenaRBW18,AlgehedR17,StefanRMM11,RussoCH08,TsaiRH07,Arden17}.
Both \citet{AlgehedR17} and MAC~\citep{VassenaRBW18}, moreover, handle effectful computation via monadic reasoning.
\citet{AlgehedR17} in particular advocate building noninterfering pure languages, and using monads to define effects on top of them.
They do not, however, explore the connections to \pc systems.
MAC~\cite{VassenaRBW18} combines the monads for effects and the monad for labels, and therefore still requires a \pc label.

\subsection{The Theory of Effects}
\label{sec:related-eff}

Type-and-effect systems originated as a program analysis technique~\citep{NielsonN99,Lucassen88,Nielson96}.
This technique allowed compilers to leverage the type system of their source language to track other properties of programs, enabling optimizations like dead-code elimination that may behave differently depending on effects.

\citet{WadlerT98} gave type-and-effect systems semantics via monads by recognizing the correspondence between type-and-effect systems and \citeposessive[Moggi91]{Moggi89} notions of computation.
This result gave rise to a long line of work describing generalizations of monads which could be used to give semantics to as many type-and-effect systems as possible.
The most relevant for this work are \citeauthor{WadlerT98}'s (and \citeposessive{OrchardPM14}) indexed monads, which work on a lattice of effects~\citep{WadlerT98} and \citeposessive{Tate13} productors, which work on an arbitrary effector.

\section{Conclusion}
\label{sec:conclusion}

We have developed a framework that gives semantics to program-counter labels based on the semantics of producer effects.
This choice supports an abstract perspective, allowing us to reason about a language \emph{feature} without being tied to a specific language.
The Noninterference Half-Off Theorem (Theorem~\ref{thm:niho}) thus proved noninterference for a large swath of languages---any language admitting the simple rules in Figures~\ref{fig:framework-rules} and~\ref{fig:pc-system-rules}.
Moreover, the proof technique the theorem suggests provides simple proofs of noninterference for important effects: state, exceptions, and nontermination.
It even applies to languages with multiple types of effects, as we saw in Section~\ref{sec:state-and-exns}.

By viewing possible nontermination as an effect, we both achieved a half-off proof of termination-sensitive noninterference and unified the treatment of termination sensitivity with that of other effects.
Previously, these had only been considered separately.
Hopefully, this new understanding of termination sensitivity will allow us to build better termination-sensitive type systems.

We also demonstrated the power of our framework by using it to formalize the folklore belief that the \pc is a lower bound on the effects in secure programs.
Surprisingly, such \pc-boundedness is \emph{not} a theorem of our semantics.
It is, however, a theorem of a slightly expanded version of our semantics.
Moreover, this extension is suggestive of a categorical construct called a \emph{distributive law}.
Exploring this connection would be interesting future work.

In fact, a categorical perspective infuses this entire work.
The semantics of effects are usually given categorically, so perhaps this is unsurprising.
Further formalizing this work categorically would require a categorical models of noninterference,
perhaps developing a connection between our semantics and \citeposessive{Kavvos19} semantics of pure noninterference given via modal types.

We believe there are two other important directions for future work.
First, our framework could influence the design of, and facilitate noninterference proofs for, a language with secure algebraic effect handlers.
Algebraic effect handlers~\citep{Leijen16,Pretnar10,PlotkinP03,BauerP15,PlotkinP09} allow programmers to specify their own effectful operations while retaining the fundamental properties of a pure language, hopefully making secure programming considerably easier.
Second, one may be able to expand our framework to other security guarantees by replacing preservation of equivalences in Figure~\ref{fig:framework-rules} with preservation of arbitrary \emph{hyperproperties}~\citep{ClarksonS10}.
Such an extension might have important applications in the design of secure programming languages in general.

The results we have developed in this paper and the future work we suggest all require the generality of our semantic framework.
We hope that future work also adopts abstract perspectives to similarly prove highly general results.

\section*{Acknowledgments}
\label{sec:acks}

This work originated from ideas we had while designing the First Order Logic for Flow-Limited Authorization
along with Pedro de Amorim, Owen Arden, and Ross Tate.
Deepak Garg and Andrew C. Myers helped us focus the work and,
along with Deian Stefan, provided guidance on how to explain our results to a broader audience.
Maximilian Algehed pointed us to important related work.
Co\c{s}ku Acay, Jonathan DiLorenzo, Matthew Milano, and Isaac Sheff helped with editing.
Finally, our shepherd, William J. Bowman, and our anonymous reviewers provided tremendously insightful comments and helpful suggestions.

This project was supported in part by a fellowship awarded through the National Defense Science and Engineering Graduate (NDSEG) Fellowship Program, sponsored by the Air Force Research Laboratory (AFRL), the Office of Naval Research (ONR), and the Army Research Office (ARO).
Any opinions, findings, conclusions, or recommendations expressed here are those of the authors and may not reflect those of these sponsors.

\bibliography{bibtex/references}

\appendix
\section{Programs with Multiple Inputs}
\label{sec:multi-input}

In Section~\ref{sec:prod-effects-ifc}, we developed a theory of effects in languages with information-flow-control types.
This was designed to be extremely general.
However, we developed our theory for single-input, single-output programs.
In this appendix, we consider expanding our theory to multiple-input programs, such as the simply-typed $\lambda$-calculus.

\citet[see Section 9]{Tate13}, mentions that a productor can be viewed as a 2-functor from an effector to the category of categories.
Thus, each effect~$\varepsilon$ is mapped to a functor~$\langle\CapturedType{\varepsilon}{-}, \MapE{\varepsilon}\rangle$, and each inequality $\ComposeE{\varepsilon_1, \ldots, \varepsilon_n}{\varepsilon}$ is mapped to a natural transformation $\CapturedType{\varepsilon_1}{-} \circ \cdots \circ \CapturedType{\varepsilon_n}{-} \Rightarrow \CapturedType{\varepsilon}{-}$.
He then suggests that we can move to multiple-input languages by defining changing the base to premonoidal categories, defining a \emph{strong} productor.

To understand this, let us define a premonoidal category~$\mathcal{C}$~\citep{Jeffrey97}:
\begin{defn}[Premonoidal Category]
A \emph{premonoidal category} is a category $\mathcal{C}$ along with
\begin{itemize}
\item A binary operation on objects, written $- \otimes -$
\item For every object $\Gamma$, two functors $\Gamma \ltimes -$ and $- \rtimes \Gamma$, such that their action on objects is $- \otimes -$
\end{itemize}
\end{defn}
This is enough to define a notion of propagating context.
We have suggestively written objects of our category as $\Gamma$, and we can define $\Gamma_1, \Gamma_2$ as $\Gamma_1 \otimes \Gamma_2$.
Then, for any morphism $\rho : \Gamma_1 \to \Gamma_2$, we can think of $\rho$ as a program that operates in an environment $\Gamma_1$, and then finishes having changed the environment to $\Gamma_2$.
Then, $\Gamma \ltimes \rho : \Gamma \otimes \Gamma_1 \to \Gamma \otimes \Gamma_2$, so we have propagated the context $\Gamma$.

We can put this in perhaps-more-familiar programming-languages terms.
A premonoidal category is one where the following rules are admissible:
\begin{mathpar}
  \infer{\Gamma_1 \proves \rho \dashv \Gamma_2}{\Gamma_1, \Gamma \proves \rho \rtimes \Gamma \dashv \Gamma_2, \Gamma} \and
  \infer{\Gamma_1 \proves \rho \dashv \Gamma_2}{\Gamma, \Gamma_1 \proves \Gamma \ltimes \rho \dashv \Gamma, \Gamma_2}
\end{mathpar}
Here, we write $\Gamma_1 \proves \rho \dashv \Gamma_2$ as the typing judgment representing a program $\rho : \Gamma_1 \to \Gamma_2$.

Then, a strong productor is simply a productor for which every functor $\CapturedType{\varepsilon}{-}$ is a premonoidal functor, that is $\CapturedType{\varepsilon}{\Gamma_1 \otimes \Gamma_2} = \CapturedType{\varepsilon}{\Gamma_1} \otimes \CapturedType{\varepsilon}{\Gamma_2}$.
In programming-language terms, the following rules are admissible:
\begin{mathpar}
  \infer{\Gamma_1 \proves \rho \dashv \Gamma_2}{\MapE{\varepsilon}(\rho \rtimes \Gamma) \progEq{(\CapturedType{\varepsilon}{\Gamma_1}, \CapturedType{\varepsilon}{\Gamma}),\,(\CapturedType{\varepsilon}{\Gamma_2}, \CapturedType{\varepsilon}{\Gamma})} \MapE{\varepsilon}(\rho) \rtimes \CapturedType{\varepsilon}{\Gamma}} \and
  \infer{\Gamma_1 \proves \rho \dashv \Gamma_2}{\MapE{\varepsilon}(\Gamma \ltimes \rho) \progEq{(\CapturedType{\varepsilon}{\Gamma}, \CapturedType{\varepsilon}{\Gamma_1}),\,(\CapturedType{\varepsilon}{\Gamma}, \CapturedType{\varepsilon}{\Gamma_2})} \CapturedType{\varepsilon}{\Gamma} \ltimes \MapE{\varepsilon}(\rho)}
\end{mathpar}

For a linear language, this is enough.
But we often want to deal with non-linear languages with the following notion of sequencing:
$$\infer{
  \Gamma \proves p_1 : t_1 \ctxsep \varepsilon_1 \\
  \Gamma, x_1 \ty t_1 \proves p_2 : t_2 \ctxsep \varepsilon_2 \\
  \cdots \\
  \Gamma, x_1 \ty t_1, \ldots, x_{n-1} \ty t_{n-1} \proves p_n : t_{n} \ctxsep \varepsilon_n \\\\
  \ComposeE{\varepsilon_1, \ldots, \varepsilon_n}{\varepsilon}
}{\Gamma \proves \LetIn*{x_1}{p_1}{\LetIn*{x_2}{p_2}{\cdots\LetIn*{x_{n-1}}{p_{n-1}}{p_n}}} : t_n \ctxsep \varepsilon}$$
In order to give semantics to this rule, we need an extra few assumptions.
In particular, we need a doubling natural transformation $\Delta_{\Gamma} : \Gamma \to \Gamma \otimes \Gamma$, which is preserved by $\MapE{\varepsilon}$.
That is, we need the following rules to be admissible:
\begin{mathpar}
  \infer{ }{\Gamma \proves \Delta_\Gamma \dashv \Gamma, \Gamma} \and
  \infer{ }{\MapE{\varepsilon}(\Delta_\Gamma) \progEq{\CapturedType{\varepsilon}{\Gamma},\,(\CapturedType{\varepsilon}{\Gamma},\CapturedType{\varepsilon}{\Gamma})} \Delta_{\CapturedType{\varepsilon}{\Gamma}}}
\end{mathpar}

This allows us to prove the following:
\begin{thm}[Semantics of Effectful Composition with Let]
  \label{thm:sem-eff-comp-let}
$$\infer{
  \Gamma \proves p_1 : t_1 \ctxsep \varepsilon_1 \\
  \cdots \\
  \Gamma, x_1 \ty t_1, \ldots, x_{n-1} \ty t_{n-1} \proves p_n : t_n \ctxsep \varepsilon_n \\
  \ComposeE{\varepsilon_1, \ldots, \varepsilon_n}{\varepsilon}
}{\Captured{\varepsilon}{\LetIn{x_1}{p_1}{\ldots\LetIn{x_{n-1}}{p_{n-1}}{p_n}}}
  \progEq{\Gamma,\,(\Gamma,x_1 \ty t_1, \ldots, x_n \ty t_n)}
  {\begin{array}{@{}l@{\;{\seq}\;}l@{}}
    \Delta_{\Gamma} & \Gamma \ltimes \Captured{\varepsilon_1}{p_1} \seq \JoinE{}{\varepsilon_1} \rtimes \CapturedType{\varepsilon}{x_1 \ty t_1} \\
    & \MapE{\varepsilon_1}
    \begin{array}[t]{@{}l@{\;{\seq}\;}l@{}}
      (\Delta_{\Gamma,x_1 \ty t_1} & \cdots \\
      & \MapE{\varepsilon_n}(\Captured{\varepsilon_n}{p_n}))
    \end{array}\\
    & \JoinE{\varepsilon_1, \ldots \varepsilon_n}{\varepsilon}
  \end{array}}}$$
\end{thm}

In order to extend to information-flow-control typed languages, we need only assume that labeling is a \emph{strong premonoidal} functor.
That is, the following rules are admissible:
\begin{mathpar}
  \LabeledType{\ell}{\Gamma_1, \Gamma_2} = \LabeledType{\ell}{\Gamma_1}, \LabeledType{\ell}{\Gamma_2} \and
  \infer{ }{\Gamma_1, \LabeledType{\ell}{\Gamma_2} \proves \textsf{str}_{\ell, \Gamma_1, \Gamma_2} \dashv \LabeledType{\ell}{\Gamma_1, \Gamma_2}}   
\end{mathpar}
Then, extending Theorem~\ref{thm:sem-eff-comp-let} to \Unlabel is not difficult.

\section{Full  Rules for  Example Languages}
\label{sec:full-pure-typing}

\subsection{DCC}

Here we present the full type system for our fragment of DCC from Section~\ref{sec:pure-ifc-lang}

{\small
\begin{mathpar}
  \infer*[left=Var]{\Gamma(x) = \tau}{\Gamma \proves x : \tau}
  \and
  \infer*[left=Unit]{ }{\Gamma \proves () : \Unit}
  \\
  \infer*[left=Pair]{
    \Gamma \proves e_1 : \tau_1 \\
    \Gamma \proves e_2 : \tau_2
  }{\Gamma \proves \Pair{e_1}{e_2} : \tau_1 \times \tau_2}
  \and
  \infer*[left=Proj]{\Gamma \proves e : \tau_1 \times \tau_2}{\Gamma \proves \Proj{i}(e) : \tau_i}
  \\
  \infer*[left=InL]{
    \Gamma \proves e : \tau_1 \\
  }{\Gamma \proves \Inl(e) : \tau_1 + \tau_2}
  \and
  \infer*[left=InR]{
    \Gamma \proves e : \tau_2
  }{\Gamma \proves \Inr(e) : \tau_1 + \tau_2}
  \and
  \infer*[left=Match]{
    \Gamma \proves e : \tau_1 + \tau_2 \\\\
    \Gamma, x \ty \tau_1 \proves e_1 : \tau \\
    \Gamma, y \ty \tau_2 \proves e_2 : \tau
  }{\Gamma \proves \MatchSum*{e}{x}{e_1}{y}{e_2} : \tau}
  \\
  \infer*[left=Lam]{
    \Gamma, x \ty \tau_1 \proves e : \tau_2 \\
  }{\Gamma \proves \lambda x \ty \tau_1.\, e : \tau_1 \to \tau_2}
  \and
  \infer*[left=App]{
    \Gamma \proves f : \tau_1 \to \tau_2 \\
    \Gamma \proves e : \tau_1
  }{\Gamma \proves f~e : \tau_2}
  \\
  \infer*[left=Label]{
    \Gamma \proves e : \tau
  }{\Gamma \proves \LabeledProgram{\ell}{e} : \LabeledType{\ell}{\tau}}
  \and
  \infer*[left=Unlabel]{
    \Gamma \proves e_1 : \LabeledType{\ell}{\tau_1} \\
    \Gamma, x \ty \tau_1 \proves e_2 : \tau_2 \\
    \ell \prot \tau_2
  }{\Gamma \proves \UnlabelAs*{e_1}{x}{e_2} : \tau_2}
\end{mathpar}}

\paragraph{Protection Rules}
{\small
  \begin{mathpar}
    \infer{\ell \flowsto \ell'}{\ell \prot \LabeledType{\ell'}{\tau}}
    \and
    \infer{\ell \prot \tau}{\ell \prot \LabeledType{\ell'}{\tau}}
    \and
    \infer{\ell \prot \tau_1 \\ \ell \prot \tau_2}{\ell \prot \tau_1 \times \tau_2}
    \and
    \infer{\ell \prot \tau_2}{\ell \prot \tau_1 \to \tau_2}
\end{mathpar}}

\paragraph{Operational Semantics}
{\small
\begin{mathpar}
  \infer{e \stepsone e'}{E[e] \stepsone E[e']} \and
  (\lambda x \ty \tau.\,e)~v \stepsone \subst{e}{x}{v} \and
  \Proj{i}(\Pair{v_1}{v_2}) \stepsone v_i \and
  \MatchSum*{\Inl(v)}{x}{e_1}{y}{e_2} \stepsone \subst{e_1}{x}{v} \and
  \MatchSum*{\Inr(v)}{x}{e_1}{y}{e_2} \stepsone \subst{e_2}{y}{v} \and
  \UnlabelAs*{\LabeledProgram{\ell}{v}}{x}{e} \stepsone \subst{e}{x}{v}
\end{mathpar}}

\subsection{DCC with State and Exceptions}

\paragraph{PC Type System}
{\small
\begin{mathpar}
  \infer*[left=Var]{\Gamma(x) = \tau}{\Gamma \ctxsep \pc \proves x : \tau}
  \and
  \infer*[left=Unit]{ }{\Gamma \ctxsep \pc \proves () : \Unit}
  \\
  \infer*[left=Pair]{
    \Gamma \ctxsep \pc \proves e_1 : \tau_1 \\
    \Gamma \ctxsep \pc \proves e_2 : \tau_2
  }{\Gamma \ctxsep \pc \proves \Pair{e_1}{e_2} : \tau_1 \times \tau_2}
  \and
  \infer*[left=Proj]{\Gamma \ctxsep \pc \proves e : \tau_1 \times \tau_2}{\Gamma \ctxsep \pc \proves \Proj{i}(e) : \tau_i}
  \\
  \infer*[left=InL]{
    \Gamma \ctxsep \pc \proves e : \tau_1 \\
  }{\Gamma \ctxsep \pc \proves \Inl(e) : \tau_1 + \tau_2}
  \and
  \infer*[left=InR]{
    \Gamma \ctxsep \pc \proves e : \tau_2
  }{\Gamma \ctxsep \pc \proves \Inr(e) : \tau_1 + \tau_2}
  \and
  \infer*[left=Match]{
    \Gamma \ctxsep \pc \proves e : \tau_1 + \tau_2 \\\\
    \Gamma, x \ty \tau_1 \ctxsep \pc \proves e_1 : \tau \\
    \Gamma, y \ty \tau_2 \ctxsep \pc \proves e_2 : \tau
  }{\Gamma \ctxsep \pc \proves \MatchSum*{e}{x}{e_1}{y}{e_2} : \tau}
  \\
  \infer*[left=Lam]{
    \Gamma, x \ty \tau_1 \ctxsep \pc_1 \proves e : \tau_2 \\
  }{\Gamma \ctxsep \pc_2 \proves \lambda x \ty \tau_1.\, e : \tau_1 \pcto[\pc_1] \tau_2}
  \and
  \infer*[left=App]{
    \Gamma \ctxsep \pc_1 \proves f : \tau_1 \pcto[\pc_2] \tau_2 \\
    \Gamma \ctxsep \pc_1 \proves e : \tau_1\\
    \pc_1 \flowsto \pc_2
  }{\Gamma \ctxsep \pc_1 \proves f~e : \tau_2}
  \\
  \infer*[left=Label]{
    \Gamma \ctxsep \pc \proves e : \tau
  }{\Gamma \ctxsep \pc \proves \LabeledProgram{\ell}{e} : \LabeledType{\ell}{\tau}}
  \and
  \infer*[left=Unlabel]{
    \Gamma \ctxsep \pc \proves e_1 : \LabeledType{\ell}{\tau_1} \\
    \Gamma, x \ty \tau_1 \ctxsep \pc \JoinL \ell \proves e_2 : \tau_2 \\
    \ell \prot \tau_2
  }{\Gamma \ctxsep \pc \proves \UnlabelAs*{e_1}{x}{e_2} : \tau_2}
  \\
  \infer*[left=Variance]{
    \Gamma \ctxsep \pc' \proves e : \tau \\
    \pc \flowsto \pc' \\
  }{\Gamma \ctxsep \pc \proves e : \tau}
  \\
  \infer*[left=Read]{ }{\Gamma \ctxsep \pc \proves \Read : \sigma}
  \and
  \infer*[left=Write]{
    \Gamma \ctxsep \pc \proves e : \sigma\\
    \pc \flowsto \StateLabel
  }{\Gamma \proves \Write(e) : \Unit}
  \\
  \infer*[left=Throw]{
    \pc \flowsto \ExnLabel
  }{\Gamma \ctxsep \pc \proves \Throw : \tau}
  \and
  \infer*[left=TryCatch]{
    \Gamma \ctxsep \pc \proves e_1 : \tau\\
    \Gamma \ctxsep \pc \proves e_2 : \tau\\
    \ExnLabel \prot \tau
  }{\Gamma \ctxsep \pc \proves \TryCatch{e_1}{e_2} : \tau}
\end{mathpar}}

\paragraph{PC Protection Rules}
{\small
\begin{mathpar}
  \infer{\ell \flowsto \ell'}{\ell \prot \LabeledType{\ell'}{\tau}}
  \and
  \infer{\ell \prot \tau}{\ell \prot \LabeledType{\ell'}{\tau}}
  \and
  \infer{\ell \prot \tau_1 \\ \ell \prot \tau_2}{\ell \prot \tau_1 \times \tau_2}
  \and
  \infer{\ell \prot \tau_2 \\ \ell \flowsto \pc}{\ell \prot \tau_1 \pcto \tau_2}
\end{mathpar}}

\paragraph{Type-and-Effect Typing Rules}
{\small
\begin{mathpar}
  \infer*[left=Var]{\Gamma(x) = \tau}{\Gamma \proves x : \tau \ctxsep \varnothing}
  \and
  \infer*[left=Unit]{ }{\Gamma \proves () : \Unit \ctxsep \varnothing}
  \\
  \infer*[left=Pair]{
    \Gamma \proves e_1 : \tau_1 \ctxsep \varepsilon_1\\
    \Gamma \proves e_2 : \tau_2 \ctxsep \varepsilon_2
  }{\Gamma \proves \Pair{e_1}{e_2} : \tau_1 \times \tau_2 \ctxsep \varepsilon_1 \cup \varepsilon_2}
  \and
  \infer*[left=Proj]{\Gamma \proves e : \tau_1 \times \tau_2 \ctxsep \varepsilon}{\Gamma \ctxsep \pc \proves \Proj{i}(e) : \tau_i \ctxsep \varepsilon}
  \\
  \infer*[left=InL]{
    \Gamma \proves e : \tau_1  \ctxsep \varepsilon\\
  }{\Gamma \proves \Inl(e) : \tau_1 + \tau_2 \ctxsep \varepsilon}
  \and
  \infer*[left=InR]{
    \Gamma \proves e : \tau_2  \ctxsep \varepsilon
  }{\Gamma \proves \Inr(e) : \tau_1 + \tau_2\ctxsep \varepsilon}
  \and
  \infer*[left=Match]{
    \Gamma \proves e : \tau_1 + \tau_2 \ctxsep \varepsilon_1\\\\
    \Gamma, x \ty \tau_1 \proves e_1 : \tau  \ctxsep \varepsilon_2\\
    \Gamma, y \ty \tau_2 \proves e_2 : \tau \ctxsep \varepsilon_2
  }{\Gamma \proves \MatchSum*{e}{x}{e_1}{y}{e_2} : \tau \ctxsep \varepsilon_1 \cup \varepsilon_2}
  \\
  \infer*[left=Lam]{
    \Gamma, x \ty \tau_1 \proves e : \tau_2 \ctxsep \varepsilon\\
  }{\Gamma  \proves \lambda x \ty \tau_1.\, e : \tau_1 \effto \tau_2 \ctxsep \varnothing}
  \and
  \infer*[left=App]{
    \Gamma \proves f : \tau_1 \effto[\varepsilon_1] \tau_2 \ctxsep \varepsilon_2 \\
    \Gamma \proves e : \tau_1 \ctxsep \varepsilon_3 \\
  }{\Gamma \proves f~e : \tau_2 \ctxsep \varepsilon_1 \cup \varepsilon_2 \cup \varepsilon_3}
  \\
  \infer*[left=Label]{
    \Gamma \proves e : \tau \ctxsep \varepsilon
  }{\Gamma \proves \LabeledProgram{\ell}{e} : \LabeledType{\ell}{\tau} \ctxsep \varepsilon}
  \and
  \infer*[left=Unlabel]{
    \Gamma \proves e_1 : \LabeledType{\ell}{\tau_1} \ctxsep \varepsilon_1 \\
    \Gamma, x \ty \tau_1 \proves e_2 : \tau_2 \ctxsep \varepsilon_2 \\\\
    \ell \prot \tau_2 \\
    \ell \flowsto \ell_{\e_2}
  }{\Gamma \proves \UnlabelAs*{e_1}{x}{e_2} : \tau_2 \ctxsep \varepsilon_1 \cup \varepsilon_2}
  \\
  \infer*[left=Variance]{
    \Gamma \proves e : \tau \ctxsep \varepsilon'\\
    \varepsilon' \subseteq \varepsilon
  }{\Gamma \proves e : \tau \ctxsep \varepsilon}
  \\
  \infer*[left=Read]{
  }{\Gamma \proves \Read : \sigma \ctxsep \REff}
  \and
  \infer*[left=Write]{
    \Gamma \proves e : \sigma \ctxsep \varepsilon\\
  }{\Gamma \proves \Write(e) : \Unit \ctxsep \varepsilon \cup \WEff}
  \\
  \infer*[left=Throw]{
  }{\Gamma \proves \Throw : \tau \ctxsep \ExnEff}
  \and
  \infer*[left=TryCatch]{
    \Gamma \proves e_1 \ty \tau \ctxsep \varepsilon_1 \cup E\\
    \Gamma \proves e_2 \ty \tau \ctxsep \varepsilon_2\\
    \ExnLabel \prot \tau
  }{\Gamma \proves \TryCatch{e_1}{e_2} : \tau \ctxsep \varepsilon_1 \cup \varepsilon_2}
\end{mathpar}}

\paragraph{Type-and-Effect Protection Rules}
{\small
\begin{mathpar}
  \infer{\ell \flowsto \ell'}{\ell \prot \LabeledType{\ell'}{\tau}}
  \and
  \infer{\ell \prot \tau}{\ell \prot \LabeledType{\ell'}{\tau}}
  \and
  \infer{\ell \prot \tau_1 \\ \ell \prot \tau_2}{\ell \prot \tau_1 \times \tau_2}
  \and
  \infer{\ell \prot \tau_2\\ \ell \flowsto \ell_\varepsilon}{\ell \prot \tau_1 \effto \tau_2}
\end{mathpar}}

\paragraph{Operational Semantics}
{\small
\begin{mathpar}
  \infer{\langle e, s \rangle \stepsone \langle e', s' \rangle}{\langle E[e], s \rangle \stepsone \langle E[e'], s' \rangle} \and
  \langle(\lambda x \ty \tau.\,e)~v, s\rangle \stepsone \langle\subst{e}{x}{v},s\rangle \and
  \langle\Proj{i}(\Pair{v_1}{v_2}), s\rangle \stepsone \langle v_i, s\rangle \and
  \langle\MatchSum*{\Inl(v)}{x}{e_1}{y}{e_2},s\rangle \stepsone \langle\subst{e_1}{x}{v},s\rangle \and
  \langle\MatchSum*{\Inr(v)}{x}{e_1}{y}{e_2},s\rangle \stepsone \langle\subst{e_2}{y}{v},s\rangle \and
  \langle\UnlabelAs*{\LabeledProgram{\ell}{v}}{x}{e},s\rangle \stepsone \langle\subst{e}{x}{v},s\rangle \\
  \langle\Read, s\rangle \stepsone \langle s, s \rangle\and
  \langle\Write(v), s\rangle \stepsone \langle (), v \rangle \\
  \langle T[\Throw], s \rangle \stepsone \langle \Throw, s\rangle\and
  \langle \TryCatch{v}{e},s \rangle \stepsone \langle v, s \rangle \and
  \langle \TryCatch{\Throw}{e}, s \rangle \stepsone \langle e, s \rangle
\end{mathpar}}

\subsection{DCC with Fixpoints and Pointed Types}

{\small
\begin{mathpar}
  \infer*[left=Var]{\Gamma(x) = \tau}{\Gamma \proves x : \tau}
  \and
  \infer*[left=Unit]{ }{\Gamma \proves () : \Unit}
  \\
  \infer*[left=Pair]{
    \Gamma \proves e_1 : \tau_1 \\
    \Gamma \proves e_2 : \tau_2
  }{\Gamma \proves \Pair{e_1}{e_2} : \tau_1 \times \tau_2}
  \and
  \infer*[left=Proj]{\Gamma \proves e : \tau_1 \times \tau_2}{\Gamma \proves \Proj{i}~e : \tau_i}
  \\
  \infer*[left=InL]{
    \Gamma \proves e : \tau_1 \\
  }{\Gamma \proves \Inl(e) : \tau_1 + \tau_2}
  \and
  \infer*[left=InR]{
    \Gamma \proves e : \tau_2
  }{\Gamma \proves \Inr(e) : \tau_1 + \tau_2}
  \and
  \infer*[left=Match]{
    \Gamma \proves e : \tau_1 + \tau_2 \\\\
    \Gamma, x \ty \tau_1 \proves e_1 : \tau \\
    \Gamma, y \ty \tau_2 \proves e_2 : \tau
  }{\Gamma \proves \MatchSum*{e}{x}{e_1}{y}{e_2} : \tau}
  \\
  \infer*[left=Lam]{
    \Gamma, x \ty \tau_1 \proves e : \tau_2 \\
  }{\Gamma \proves \lambda x \ty \tau_1.\, e : \tau_1 \to \tau_2}
  \and
  \infer*[left=App]{
    \Gamma \proves f : \tau_1 \to \tau_2 \\
    \Gamma \proves e : \tau_1
  }{\Gamma \proves f~e : \tau_2}
  \\
  \infer*[left=Label]{
    \Gamma \proves e : \tau
  }{\Gamma \proves \LabeledProgram{\ell}{e} : \LabeledType{\ell}{\tau}}
  \and
  \infer*[left=Unlabel]{
    \Gamma \proves e_1 : \LabeledType{\ell}{\tau_1} \\
    \Gamma, x \ty \tau_1 \proves e_2 : \tau_2 \\
    \ell \prot \tau_2
  }{\Gamma \proves \UnlabelAs*{e_1}{x}{e_2} : \tau_2}
  \and
  \infer*[left=Fix]{
    \Gamma, f \ty \tau \proves e : \tau\\
    \proves \tau~\Ptd
  }{\Gamma \proves \Fixpoint{f}{\tau}{e} : \tau}
  \and
  \infer*[left=Lift]{
    \Gamma \proves e : \tau
  }{\Gamma \proves \LiftProgram{e} : \LiftType{\tau}}
  \and
  \infer*[left=Seq]{
    \Gamma \proves e_1 : \LiftType{\tau_1}\\
    \Gamma, x \ty \tau_1 \proves e_2 : \tau_2\\\\
    \proves \tau_2~\Ptd
  }{\Gamma \proves \SeqTerm*{x}{e_1}{e_2} : \tau_2}
\end{mathpar}}

\paragraph{Protection Rules}
{\small
\begin{mathpar}
  \infer{\ell \flowsto \ell'}{\ell \prot \LabeledType{\ell'}{\tau}}
  \and
  \infer{\ell \prot \tau}{\ell \prot \LabeledType{\ell'}{\tau}}
  \and
  \infer{\ell \prot \tau_1 \\ \ell \prot \tau_2}{\ell \prot \tau_1 \times \tau_2}
  \and
  \infer{\ell \prot \tau_2}{\ell \prot \tau_1 \to \tau_2}
  \and
  \infer{\ell \flowsto \TermLabel\\ \ell \prot \tau}{\ell \prot \LiftType{\tau}}
\end{mathpar}}

\paragraph{Operational Semantics}
{\small
\begin{mathpar}
  \infer{e \stepsone e'}{E[e] \stepsone E[e']} \and
  (\lambda x \ty \tau.\,e)~v \stepsone \subst{e}{x}{v} \and
  \Proj{i}((v_1, v_2)) \stepsone v_i \and
  \MatchSum*{\Inl(v)}{x}{e_1}{y}{e_2} \stepsone \subst{e_1}{x}{v} \and
  \MatchSum*{\Inr(v)}{x}{e_1}{y}{e_2} \stepsone \subst{e_2}{y}{v} \and
  \UnlabelAs*{\LabeledProgram{\ell}{v}}{x}{e} \stepsone \subst{e}{x}{v} \and
  \Fixpoint{f}{\tau}{e} \stepsone \subst{e}{f}{\Fixpoint{f}{\tau}{e}} \and
  \SeqTerm*{x}{\LiftProgram{v}}{e} \stepsone \subst{e}{x}{v}
\end{mathpar}}

\subsection{DCC with Fixpoints: PC and Type-and-Effect Systems}

\paragraph{PC Type System}
{\small
\begin{mathpar}
  \infer*[left=Var]{\Gamma(x) = \tau}{\Gamma \ctxsep \pc \proves x : \tau}
  \and
  \infer*[left=Unit]{ }{\Gamma \ctxsep \pc \proves () : \Unit}
  \\
  \infer*[left=Pair]{
    \Gamma \ctxsep \pc \proves e_1 : \tau_1 \\
    \Gamma \ctxsep \pc \proves e_2 : \tau_2
  }{\Gamma \ctxsep \pc \proves \Pair{e_1}{e_2} : \tau_1 \times \tau_2}
  \and
  \infer*[left=Proj]{\Gamma \ctxsep \pc \proves e : \tau_1 \times \tau_2}{\Gamma \ctxsep \pc \proves \Proj{i}~e : \tau_i}
  \\
  \infer*[left=InL]{
    \Gamma \ctxsep \pc \proves e : \tau_1 \\
  }{\Gamma \ctxsep \pc \proves \Inl(e) : \tau_1 + \tau_2}
  \and
  \infer*[left=InR]{
    \Gamma \ctxsep \pc \proves e : \tau_2
  }{\Gamma \ctxsep \pc \proves \Inr(e) : \tau_1 + \tau_2}
  \and
  \infer*[left=Match]{
    \Gamma \ctxsep \pc \proves e : \tau_1 + \tau_2 \\\\
    \Gamma, x \ty \tau_1 \ctxsep \pc \proves e_1 : \tau \\
    \Gamma, y \ty \tau_2 \ctxsep \pc \proves e_2 : \tau
  }{\Gamma \ctxsep \pc \proves \MatchSum*{e}{x}{e_1}{y}{e_2} : \tau}
  \\
  \infer*[left=Lam]{
    \Gamma, x \ty \tau_1 \ctxsep \pc_1 \proves e : \tau_2 \\
  }{\Gamma \ctxsep \pc_2 \proves \lambda x \ty \tau_1.\, e : \tau_1 \pcto[\pc_1] \tau_2}
  \and
  \infer*[left=App]{
    \Gamma \ctxsep \pc_1 \proves f : \tau_1 \pcto[\pc_2] \tau_2 \\
    \Gamma \ctxsep \pc_1 \proves e : \tau_1\\
    \pc_1 \flowsto \pc_2
  }{\Gamma \ctxsep \pc_1 \proves f~e : \tau_2}
  \\
  \infer*[left=Label]{
    \Gamma \ctxsep \pc \proves e : \tau
  }{\Gamma \ctxsep \pc \proves \LabeledProgram{\ell}{e} : \LabeledType{\ell}{\tau}}
  \and
  \infer*[left=Unlabel]{
    \Gamma \ctxsep \pc \proves e_1 : \LabeledType{\ell}{\tau_1} \\
    \Gamma, x \ty \tau_1 \ctxsep \pc \JoinL \ell \proves e_2 : \tau_2 \\
    \ell \prot \tau_2
  }{\Gamma \ctxsep \pc \proves \UnlabelAs*{e_1}{x}{e_2} : \tau_2}
 \\
 \infer*[left=Variance]{
   \Gamma \ctxsep \pc' \proves e : \tau \\
   \pc \flowsto \pc'
 }{\Gamma \ctxsep \pc \proves e : \tau}
 \and
 \infer*[left=Fix]{
   \Gamma, f \ty \tau \ctxsep \pc \proves e : \tau\\
   \pc \flowsto \TermLabel
 }{\Gamma \ctxsep \pc \proves \Fixpoint{f}{\tau}{e} : \tau}
 \end{mathpar}}

\paragraph{PC Protection Rules}
{\small
\begin{mathpar}
  \infer{\ell \flowsto \ell'}{\ell \prot \LabeledType{\ell'}{\tau}}
  \and
  \infer{\ell \prot \tau}{\ell \prot \LabeledType{\ell'}{\tau}}
  \and
  \infer{\ell \prot \tau_1 \\ \ell \prot \tau_2}{\ell \prot \tau_1 \times \tau_2}
  \and
  \infer{\ell \prot \tau_2 \\ \ell \flowsto \pc}{\ell \prot \tau_1 \pcto \tau_2}
\end{mathpar}}

\paragraph{Type-and-Effect Type System}
{\small
\begin{mathpar}
  \infer*[left=Var]{\Gamma(x) = \tau}{\Gamma \proves x : \tau \ctxsep \varnothing}
  \and
  \infer*[left=Unit]{ }{\Gamma \proves () : \Unit \ctxsep \varnothing}
  \\
  \infer*[left=Pair]{
    \Gamma \proves e_1 : \tau_1 \ctxsep \varepsilon_1\\
    \Gamma \proves e_2 : \tau_2 \ctxsep \varepsilon_2
  }{\Gamma \proves \Pair{e_1}{e_2} : \tau_1 \times \tau_2 \ctxsep \varepsilon_1 \cup \varepsilon_2}
  \and
  \infer*[left=Proj]{\Gamma \proves e : \tau_1 \times \tau_2 \ctxsep \varepsilon}{\Gamma \ctxsep \pc \proves \Proj{i}~e : \tau_i \ctxsep \varepsilon}
  \\
  \infer*[left=InL]{
    \Gamma \proves e : \tau_1  \ctxsep \varepsilon\\
  }{\Gamma \proves \Inl(e) : \tau_1 + \tau_2 \ctxsep \varepsilon}
  \and
  \infer*[left=InR]{
    \Gamma \proves e : \tau_2  \ctxsep \varepsilon
  }{\Gamma \proves \Inr(e) : \tau_1 + \tau_2\ctxsep \varepsilon}
  \and
  \infer*[left=Match]{
    \Gamma \proves e : \tau_1 + \tau_2 \ctxsep \varepsilon_1\\\\
    \Gamma, x \ty \tau_1 \proves e_1 : \tau  \ctxsep \varepsilon_2\\
    \Gamma, y \ty \tau_2 \proves e_2 : \tau \ctxsep \varepsilon_2
  }{\Gamma \proves \MatchSum*{e}{x}{e_1}{y}{e_2} : \tau \ctxsep \varepsilon_1 \cup \varepsilon_2}
  \\
  \infer*[left=Lam]{
    \Gamma, x \ty \tau_1 \proves e : \tau_2 \ctxsep \varepsilon\\
  }{\Gamma  \proves \lambda x \ty \tau_1.\,e : \tau_1 \effto \tau_2 \ctxsep \varnothing}
  \and
  \infer*[left=App]{
    \Gamma  \proves f : \tau_1 \effto[\varepsilon_1] \tau_2 \ctxsep \varepsilon_2 \\
    \Gamma \proves e : \tau_1 \ctxsep \varepsilon_3 \\
  }{\Gamma \proves f~e : \tau_2 \ctxsep \varepsilon_1 \cup \varepsilon_2 \cup \varepsilon_3}
  \\
  \infer*[left=Label]{
    \Gamma \proves e : \tau \ctxsep \varepsilon
  }{\Gamma \proves \LabeledProgram{\ell}{e} : \LabeledType{\ell}{\tau} \ctxsep \varepsilon}
  \and
  \infer*[left=Unlabel]{
    \Gamma \proves e_1 : \LabeledType{\ell}{\tau_1} \ctxsep \varepsilon_1 \\
    \Gamma, x \ty \tau_1 \proves e_2 : \tau_2 \ctxsep \varepsilon_2 \\\\
    \ell \prot \tau_2 \\
    \ell \flowsto \ell_{\e_2}
  }{\Gamma \proves \UnlabelAs*{e_1}{x}{e_2} : \tau_2 \ctxsep \varepsilon_1 \cup \varepsilon_2}
 \\
 \infer*[left=Variance]{
   \Gamma \proves e : \tau \ctxsep \varepsilon'\\
   \varepsilon' \subseteq \varepsilon
 }{\Gamma \proves e : \tau \ctxsep \varepsilon}
 \and
 \infer*[left=Fix]{
   \Gamma, f : \tau \proves e : \tau \ctxsep \varepsilon
 }{\Gamma \proves \Fixpoint{f}{\tau}{e} : \tau \ctxsep \PntEff}
 \end{mathpar}}

\paragraph{Type-and-Effect Protection Rules}
{\small
\begin{mathpar}
  \infer{\ell \flowsto \ell'}{\ell \prot \LabeledType{\ell'}{\tau}}
  \and
  \infer{\ell \prot \tau}{\ell \prot \LabeledType{\ell'}{\tau}}
  \and
  \infer{\ell \prot \tau_1 \\ \ell \prot \tau_2}{\ell \prot \tau_1 \times \tau_2}
  \and
  \infer{\ell \prot \tau_2\\ \ell \flowsto \ell_\varepsilon}{\ell \prot \tau_1 \effto \tau_2}
\end{mathpar}}

\paragraph{Operational Semantics}
{\small
\begin{mathpar}
  \infer{e \stepsone e'}{E[e] \stepsone E[e']} \and
  (\lambda x \ty \tau.\,e)~v \stepsone \subst{e}{x}{v} \and
  \Proj{i}((v_1, v_2)) \stepsone v_i \and
  \MatchSum*{\Inl(v)}{x}{e_1}{y}{e_2} \stepsone \subst{e_1}{x}{v} \and
  \MatchSum*{\Inr(v)}{x}{e_1}{y}{e_2} \stepsone \subst{e_2}{y}{v} \and
  \UnlabelAs*{\LabeledProgram{\ell}{v}}{x}{e} \stepsone \subst{e}{x}{v} \and
  \Fixpoint{f}{\tau}{e} \stepsone \subst{e}{f}{\Fixpoint{f}{\tau}{e}}
\end{mathpar}}

\section{Full Translations of Example Languages}
\label{sec:full-transl-example}

\subsection{Type Translations}
\label{sec:full-type-transl}

Here we define the type translations between the \pc system and the type-and-effect systems from Sections~\ref{sec:state-and-exns}~and~\ref{sec:tsni}.

The translation for the language with state and exceptions where $\Eff = 2^{\{\REff, \WEff, \ExnEff\}}$ is as follows.
\begin{align*}
  \effTypeTrans{\Unit} & = \Unit \\
  \effTypeTrans{\tau_1 + \tau_2} & = \effTypeTrans{\tau_1} + \effTypeTrans{\tau_2} \\
  \effTypeTrans{\tau_1 \times \tau_2} & = \effTypeTrans{\tau_1} \times \effTypeTrans{\tau_2} \\
  \effTypeTrans{\tau_1 \pcto \tau_2} & = \begin{cases}
    \effTypeTrans{\tau_1} \effto[\REff\WEff\ExnEff] \effTypeTrans{\tau_2} & \text{if } \pc \flowsto \ExnLabel \\
    \effTypeTrans{\tau_1} \effto[\REff\WEff] \effTypeTrans{\tau_2} & \text{if } \pc \nflowsto \ExnLabel \text{ but } \pc \flowsto \StateLabel \\
    \effTypeTrans{\tau_1} \effto[\REff] \effTypeTrans{\tau_2} & \text{otherwise}
  \end{cases} \\
  \effTypeTrans{\LabeledType{\ell}{\tau}} & = \LabeledType{\ell}{\effTypeTrans{\tau}}
\end{align*}

The translation for the potentially nonterminating language where $\Eff = \{\varnothing, \PntEff\}$ is as follows.
\begin{align*}
  \effTypeTrans{\Unit} & = \Unit \\
  \effTypeTrans{\tau_1 + \tau_2} & = \effTypeTrans{\tau_1} + \effTypeTrans{\tau_2} \\
  \effTypeTrans{\tau_1 \times \tau_2} & = \effTypeTrans{\tau_1} \times \effTypeTrans{\tau_2} \\
  \effTypeTrans{\tau_1 \pcto \tau_2} & = \begin{cases}
    \effTypeTrans{\tau_1} \effto[\PntEff] \effTypeTrans{\tau_2} & \text{if } \pc \flowsto \TermLabel \\
    \effTypeTrans{\tau_1} \effto[\varnothing] \effTypeTrans{\tau_2} & \text{otherwise}
  \end{cases} \\
  \effTypeTrans{\LabeledType{\ell}{\tau}} & = \LabeledType{\ell}{\effTypeTrans{\tau}}
\end{align*}

\subsection{Effectful DCC}
\label{sec:effectful-dcc-transl}

In Sections~\ref{sec:state-and-exns}~and~\ref{sec:tsni} our translation made use of indexed monads

\begingroup
\allowdisplaybreaks
\addtocounter{numlevels}{1}

Here we define the full translation between effectful and pure languages for our examples.
For simplicity, we use the syntactic sugar $\LetIn*{x}{e_1}{e_2}$ and $\LetIn*{\Pair{x}{y}}{e_1}{e_2}$ with their standard meanings.

As the translation for Section~\ref{sec:state-and-exns} is based on an indexed monad~\citep{WadlerT98,OrchardPM14},
we define a type-directed translation making use of $\eta_\e$, $\BindM{\e}$, and $\CoerceE{\e_1}{\e_2}$ operations for our monadic forms.
For notational clarity, we use a different but equivalent type for $\BindM{}$ than we did in Section~\ref{sec:state-and-exns}.
Specifically, we use $\BindM{\e} : \Monad{\e}(\tau_1) \to (\tau_1 \to \Monad{\e}(\tau_2)) \to \Monad{\e}(\tau_2)$.

We start by providing a general translation for the fragment of DCC that appears in both Sections~\ref{sec:state-and-exns}~and~\ref{sec:tsni},
and then discuss the specifics of each translation including their monads and the operations that appear in only one.

All of the translations are type-directed.
We use $\Captured{\e}{e}$ as shorthand to denote the translation of a derivation of $\Gamma \proves e : \tau \ctxsep \e$.

\subsection{Common Language Fragment}
\label{sec:general-translation}

The type directed translation is as follows.
\begin{align*}
  \Captured*{\infer{\Gamma(x) = \tau}{\Gamma \proves x : \tau \ctxsep \varnothing}} & = x \\[0.5em]
  \Captured*{\infer{ }{\Gamma \proves () : \Unit \ctxsep \varnothing}} & = () \\[0.5em]
  \Captured*{\infer{\Gamma \proves e : \tau_1 \ctxsep \e}{\Gamma \proves \Inl~e : \tau_1 + \tau_2 \ctxsep \e}} & =
    \BindM{\e} ~ \Captured{\e}{e} ~ (\lambda x\ty\tau_1.\, \eta_\e~(\Inl~x)) \\[0.5em]
  \Captured*{\infer{\Gamma \proves e : \tau_2 \ctxsep \e}{\Gamma \proves \Inr~e : \tau_1 + \tau_2 \ctxsep \e}} & =
    \BindM{\e} ~ \Captured{\e}{e} ~ (\lambda x\ty\tau_2.\, \eta_\e~(\Inr~x)) \\[0.5em]
  \addtocounter{numlevels}{-1}
  \Captured*{
    \infer{
      \Gamma \proves e : \tau_1 + \tau_2 \ctxsep \e_1 \\\\
      \Gamma, x\ty\tau_1 \proves e_1 : \tau \ctxsep \e_2 \\\\
      \Gamma, y\ty\tau_2 \proves e_2 : \tau \ctxsep \e_2 \\\\
      \e_1 \cup \e_2 \subseteq \e
    }{\Gamma \proves \MatchSum{e}{x}{e_1}{y}{e_2} : \tau \ctxsep \e}} & =
  \addtocounter{numlevels}{1}
    \begin{array}{@{}l@{~}l@{}}
      \BindM{\e} & (\CoerceE{\e_1}{\e}~\Captured{\e_1}{e}) \\
      & \left(\begin{array}{@{}c@{}}\lambda z\ty\tau_1 + \tau_2.\, \CoerceE{\e_2}{\e}~\MatchSum{z}{x}{\Captured{\e_2}{e_1}}{y}{\Captured{\e_2}{e_2}}\end{array}\right)
    \end{array}
  \\[0.5em]
  \Captured*{\infer{\Gamma \proves e_1 : \tau_1 \ctxsep \e_1 \\\\ \Gamma \proves e_2 : \tau_2 \ctxsep \e_2 \\\\ \e_1 \cup \e_2 \subseteq \e}{\Gamma \proves \Pair{e_1}{e_2} : \tau_1 \times \tau_2 \ctxsep \e}} & =
    \begin{array}{@{}l@{~}l@{}}
      \BindM{\e} & (\CoerceE{\e_1}{\e}~\Captured{\e_1}{e_1}) \\
      & (\lambda v_1\ty\tau_1.\,\begin{array}[t]{@{}l@{~}l@{}}\BindM{\e} & (\CoerceE{\e_2}{\e}~\Captured{\e_2}{e_2}) \\ & (\lambda v_2\ty\tau_2.\, \eta_\e~\Pair{v_1}{v_2}))\end{array}
    \end{array}\\[0.5em]
  \Captured*{\infer{\Gamma \proves e : \tau_1 \times \tau_2 \ctxsep \e}{\Gamma \proves \Proj{i}~e : \tau_i \ctxsep \e}} & =
    \BindM{\e} ~ \Captured{\e}{e} ~ (\lambda x\ty\tau.\, \eta_\e~(\Proj{i}~x)) \\[0.5em]
  \Captured*{\infer{\Gamma, x\ty\tau_1 \proves e : \tau_2 \ctxsep \e}{\Gamma \proves \lambda x\ty\tau_1.\, e : \tau_1 \effto \tau_2 \ctxsep \varnothing}} & = \lambda x\ty\tau_1.\Captured{\e}{e} \\[0.5em]
  \Captured*{\infer{\Gamma \proves e_1 : \tau_1 \effto[\e_1] \tau_2 \ctxsep \e_2 \\\\ \Gamma \proves e_2 : \tau_2 \ctxsep \e_3 \\\\ \e_1 \cup \e_2 \cup \e_3 \subseteq \e}{\Gamma \proves e_1~e_2 : \tau_2 \ctxsep \e}} & =
    \begin{array}{@{}l@{~}l@{}}
      \BindM{\e} & (\CoerceE{\e_2}{\e}~\Captured{\e_2}{e_1}) \\
      & \begin{array}[t]{@{}l@{}}
        (\lambda f\ty \tau_1 \to \Monad{\e_1}(\tau_2). \\
        \quad
        \begin{array}[t]{@{}l@{~}l@{}}
          \BindM{\e} & (\CoerceE{\e_3}{\e}~\Captured{\e_3}{e_2}) \\
          & (\lambda x\ty\tau_1.\, \CoerceE{\e_1}{\e}~(f~x)))
        \end{array}
      \end{array}
    \end{array} \\[0.5em]
  \Captured*{\infer{\Gamma \proves e : \tau \ctxsep \e}{\Gamma \proves \LabeledProgram{\ell}{e} : \LabeledType{\ell}{\tau} \ctxsep \e}} & =
    \BindM{\e} ~ \Captured{\e}{e} ~ (\lambda x\ty\tau.\, \eta_\e~(\LabeledProgram{\ell}{x})) \\[0.5em]
  \Captured*{\infer{\Gamma \proves e_1 : \LabeledType{\ell}{\tau_1} \ctxsep \e_1 \\\\ \Gamma, x\ty\tau_1 \proves e_2 : \tau \ctxsep \e_2 \\\\ \ell \prot \tau \\ \ell \flowsto \ell_{\e_2} \\\\ \e_1 \cup \e_2 \subseteq \e}{\Gamma \proves \UnlabelAs*{e_1}{x}{e_2} : \tau \ctxsep \e}} & =
    \begin{array}{@{}l@{~}l@{}}
      \BindM{\e} & (\CoerceE{\e_1}{\e} ~ \Captured{\e_1}{e_1}) \\
      & (\lambda v\ty\LabeledType{\ell}{\tau_1}.\, \CoerceE{\e_2}{\e} ~ (\UnlabelAs*{v}{x}{\Captured{\e_2}{e_2}}))
    \end{array} \\[0.5em]
  \Captured*{\infer{\Gamma \proves e : \tau \ctxsep \e' \\ \e' \subseteq \e}{\Gamma \proves e : \tau \ctxsep \e}} & = \CoerceE{\e'}{\e}~\Captured{\e'}{e}
\end{align*}

\subsection{State and Exceptions}

For this we first define the $\eta$ and $\BindM{}$ rules for each component of our indexed monad.
Note that we omit the trivial monad $\Monad{\varnothing}$ because it is the identity.

\paragraph{Exception}
\begin{mathpar}
  \eta_\ExnEff~v = \LabeledProgram{\ExnLabel}{\Inr~v}
  \and
  \begin{array}{@{}l@{}}
    \BindM{\ExnEff}~v~f = \UnlabelAs{v}{x}{
      \MatchSum{x}{\_}{\LabeledProgram{\ExnLabel}{\Inl~()}}{y}{f~y}}
  \end{array}
\end{mathpar}

\paragraph{Read}
\begin{mathpar}
  \eta_\REff~v = \lambda s\ty\sigma.\, v
  \and
  \BindM{\REff}~v~f = \lambda s\ty\sigma.\, f~(v~s)
\end{mathpar}

\paragraph{Read and Write}
Note that this also applies to \WEff effects, as we use the same monad.
\begin{mathpar}
  \eta_{\REff\WEff}~v = \lambda s\ty\sigma.\, \Pair{v}{s}
  \and
  \BindM{\REff\WEff}~v~f = \lambda s\ty\sigma.\, \LetIn*{\Pair{x}{s'}}{v~s}{f~x~s'}
\end{mathpar}

\paragraph{Read and Exception}
\begin{mathpar}
  \eta_{\REff\ExnEff}~v = \lambda s\ty\sigma.\, \LabeledProgram{\ExnLabel}{\Inr~(v~s)}
  \and
  \begin{array}{@{}l@{}}
    \BindM{\REff\ExnEff}~v~f = \lambda s\ty\sigma.\, \UnlabelAs{(v~s)}{x}{
      \MatchSum{x}{\_}{\LabeledProgram{\ExnLabel}{\Inl~()}}{y}{f~y~s}}
  \end{array}
\end{mathpar}

\paragraph{Read, Write, and Exception}
Note that this also applies to $\{\WEff,\ExnEff\}$, as we use the same monad.
\[
  \eta_{\REff\WEff\ExnEff}~v = \lambda s\ty\sigma.\,\Pair{\LabeledProgram{\ExnLabel}{\Inr~v}}{s}
  \hspace{1em}
  \begin{array}{l}
    \BindM{\REff\WEff\ExnEff}~v~f = \lambda s\ty\sigma.\, \LetIn{\Pair{x}{s'}}{v~s}{
      \UnlabelAs{x}{y}{
        \MatchSum{y}{\_}{\Pair{\LabeledProgram{\ExnLabel}{\Inl~()}}{s'}}{z}{f~z~s'}}}
  \end{array}
\]

We now define the various coercions used in our translation.
First we note that, for all effects~$\e$, $\CoerceE{\varnothing}{\e} = \eta_\e$ and $\CoerceE{\e}{\e}$ is the identity.
The rest are defined as follows:
\begin{align*}
  \CoerceE{\REff}{\REff\WEff}~v & = \lambda s\ty\sigma.\, \Pair{v~s}{s} \\[0.5em]
  \CoerceE{\REff}{\REff\ExnEff}~v & = \lambda s\ty\sigma.\, \eta_\ExnEff~(v~s) \\[0.5em]
  \CoerceE{\REff}{\REff\WEff\ExnEff}~v & = \lambda s\ty\sigma.\, \Pair{\eta_\ExnEff~(v~s)}{s} \\[0.5em]
  \CoerceE{\ExnEff}{\REff\ExnEff}~v & = \lambda s\ty\sigma.\, v \\[0.5em]
  \CoerceE{\ExnEff}{\REff\WEff\ExnEff}~v & = \lambda s\ty\sigma.\, \Pair{v}{s} \\[0.5em]
  \CoerceE{\REff\WEff}{\REff\WEff\ExnEff}~v & = \lambda s\ty\sigma.\, \LetIn*{\Pair{x}{s'}}{v~s}{\Pair{\eta_\ExnEff~x}{s'}} \\[0.5em]
  \CoerceE{\REff\ExnEff}{\REff\WEff\ExnEff}~v & = \lambda s\ty\sigma.\, \Pair{v~s}{s} 
\end{align*}

Finally, we provide the translations for the four stateful operations defined in Section~\ref{sec:state-and-exns}.
We note that try-catch operates differently depending on the effects of the expression in the try block, so we provide several different translations.
\begin{align*}
  \Captured*{\infer{ }{\Gamma \proves \Read : \sigma \ctxsep \REff}} & = \lambda s\ty\sigma.\, s \\[0.5em]
  \Captured*{\infer{\Gamma \proves e : \sigma \ctxsep \e}{\Gamma \proves \Write(\e) : \Unit \ctxsep \e \cup \WEff}} & =
    \begin{array}[t]{@{}l@{~}l@{}}
      \BindM{\e \cup \WEff} & (\CoerceE{\e}{\e \cup \WEff} ~ \Captured{\e}{e}) \\
      & (\lambda x\ty\sigma.\, \CoerceE{\WEff}{\e \cup \WEff}~\Pair{()}{x})
    \end{array} \\[0.5em]
  \Captured*{\infer{ }{\Gamma \proves \Throw : \tau \ctxsep \ExnEff}} & = \LabeledProgram{\ExnLabel}{\Inl~()} \\[0.5em]
  \Captured*{\infer{\Gamma \proves e_1 : \tau \ctxsep \ExnEff \\ \Gamma \proves e_2 : \tau \ctxsep \e \\ \ExnLabel \prot \tau}{\Gamma \proves \TryCatch{e_1}{e_2} : \tau \ctxsep \e}} & =
    \UnlabelAs{\Captured{\ExnEff}{e_1}}{x}{
      \MatchSum{x}{\_}{\Captured{\e}{e_2}}{v}{\eta_\e~v}} \\
  \Captured*{\infer{\Gamma \proves e_1 : \tau \ctxsep \REff\ExnEff \\ \Gamma \proves e_2 : \tau \ctxsep \e \\ \ExnLabel \prot \tau}{\Gamma \proves \TryCatch{e_1}{e_2} : \tau \ctxsep \e \cup \REff}} & =
    \lambda s\ty\sigma.\,
      \UnlabelAs{\Captured{\REff\ExnEff}{e_1}~s}{x}{
        \MatchSum{x}{\_}{\CoerceE{\e}{\e \cup \REff}~\Captured{\e}{e_2}~s}{v}{\eta_{\e\cup \REff}~v~s}} \\
  \Captured*{\infer{\Gamma \proves e_1 : \tau \ctxsep \REff\WEff\ExnEff \\ \Gamma \proves e_2 : \tau \ctxsep \e \\ \ExnLabel \prot \tau}{\Gamma \proves \TryCatch{e_1}{e_2} : \tau \ctxsep \e \cup \REff\WEff}} & =
    \lambda s\ty\sigma.\,
      \LetIn{\Pair{x}{s'}}{\Captured{\REff\WEff\ExnEff}{e_1}~s}{
        \UnlabelAs{x}{y}{
          \MatchSum{y}{\_}{\CoerceE{\e}{\e \cup \REff\WEff}~\Captured{\e}{e_2}~s'}{v}{\eta_{\e \cup \REff\WEff}~v~s'}}} \\
\end{align*}

\subsection{Potential Nontermination}

The translation for Section~\ref{sec:tsni} is considerably simpler because it uses only a single nontermination monad.
Moreover, the monadic operations are provided directly in the language as follows:
\begin{mathpar}
  \eta~v = \LiftProgram{v}
  \and
  \BindM{}~v~f = \left(\SeqTerm*{x}{v}{f~x}\right)
\end{mathpar}

The only stateful operation not present in the base language (see Appendix~\ref{sec:general-translation}) is \Fix, which has the following translation:
$$\Captured*{\infer{\Gamma, f\ty\tau \proves e : \tau \ctxsep \e}{\Gamma \proves \Fixpoint{f}{\tau}{e} : \tau \ctxsep \PntEff}} =
\Fixpoint{f}{\LiftType{\tau}}{\CoerceE{\e}{\PntEff}~\Captured{\e}{e}}$$
Note that, since $\e$ is either $\varnothing$ or $\PntEff$, $\CoerceE{\e}{\PntEff}$ is either $\eta_\PntEff$ or the identity, respectively.

\addtocounter{numlevels}{-1}
\endgroup

\end{document}